\documentclass[]{interact}

\usepackage{epstopdf}
\usepackage[caption=false]{subfig}
\usepackage{bm}
\usepackage[ruled,vlined]{algorithm2e}
\usepackage{xcolor}
\usepackage{hyperref}
\usepackage[normalem]{ulem}

\usepackage[numbers,sort&compress]{natbib}
\bibpunct[, ]{[}{]}{,}{n}{,}{,}
\renewcommand\bibfont{\fontsize{10}{12}\selectfont}

\theoremstyle{plain}
\newtheorem{theorem}{Theorem}
\newtheorem{lemma}[theorem]{Lemma}

\theoremstyle{definition}
\newtheorem{definition}[theorem]{Definition}

\theoremstyle{remark}
\newtheorem{remark}{Remark}

\newcommand{\e}{{\varepsilon}}

\newcommand{\bdiag}{{\rm bdiag}}
\newcommand{\diag}{{\rm diag}}

\begin{document}


\title{Power Allocation Algorithms for \\ Massive MIMO Systems with Multi-Antenna Users}

\author{
\name{Evgeny~Bobrov\textsuperscript{a,b}\thanks{
Emails: eugenbobrov@ya.ru, roborisor@gmail.com, kuznetsov.victor@huawei.com, minenkov.ds@gmail.com, d.yudakov43@gmail.com}, 
Boris~Chinyaev\textsuperscript{a,b},
Viktor~Kuznetsov\textsuperscript{b}, \\
Dmitrii~Minenkov\textsuperscript{a,c},
Daniil~Yudakov\textsuperscript{a,b}
}
\affil{\textsuperscript{a} M.V. Lomonosov Moscow State University;  \\
\textsuperscript{b} %
Huawei Technologies, Russian Research Institute, Moscow Research Center \\
\textsuperscript{c} A. Ishlinsky Institute for Problems in Mechanics RAS}}

\maketitle

\begin{abstract}
Modern 5G wireless cellular networks use massive multiple-input multiple-output (MIMO) technology. This concept entails using an antenna array at a base station to concurrently service many mobile devices that have several antennas on their side. In this field, a significant role is played by the precoding (beamforming) problem. During downlink, an important part of precoding is the power allocation problem that distributes power between transmitted symbols. In this paper, we consider the power allocation problem for a class of precodings that asymptotically work as regularized zero-forcing. Under some realistic assumptions, we simplify the spectral efficiency functional and obtain tractable expressions for it. We prove that equal power allocation provides optimum for the simplified functional with total power constraint (TPC). We propose low-complexity Intersection Methods (IM) that improve equal power allocation in the case of per-antenna power constraints (PAPC). On simulations using Quadriga, the proposed IM method in combination with widely-studied Water Filling (WF) shows a significant gain in spectral efficiency while using a similar computing time as the reference Equal Power (EP) solution. 
\end{abstract}


\section{Introduction}

The massive multiple-input multiple-output (MIMO) systems have attracted a lot of attention in both academia and industry since their first appearance~\cite{5G,marzetta2010noncooperative}. The main characteristic of the massive MIMO system is the large-scale antenna arrays at the cellular base station (BS). Using a large number of antennas, the massive MIMO system can exceed the achievable rate of a conventional MIMO~\cite{ge2016multi} system and simultaneously serves (with low power consumption) several users. 

A critical issue for improving the performance of wireless networks is the efficient management of available radio resources~\cite{le2007multihop}. Numerous works are dedicated to optimal allocation of the radio resources, for example, power and bandwidth to improve the performance of wireless networks~\cite{phan2009power}. 

An important part of signal processing in downlink is precoding since with this procedure we can focus transmission signal energy on smaller regions, which allows achieving greater spectral efficiency with lower transmitted power~\cite{EE}. Various linear precodings allow directing the maximum amount of energy to the user like Maximum Ratio Transmission (MRT) or completely get rid of inter-user interference like Zero-Forcing (ZF)~\cite{ZF_MRT, RZF2}. The precoding problem is well-studied (see e.g., overview~\cite{Survey2017,Survey2015, EZF19} and textbooks~\cite{Bjornson_tb_17, Tse_tb_05} and bibliography within), nonetheless there are open questions. 
For example, most of the works consider the total power constraint (TPC) (see e.g.,~\cite{boccardi2006optimum}), the more realistic per-antenna power constraints (PAPC) are much less studied (see e.g., ~\cite{yu2006uplink,bjornson2013optimal}).

An important component of the precoding procedure is the power allocation (PA) problem that is widely discussed in the literature. In~\cite{deng2005power}, by using either the signal-to-interference-and-noise ratio (SINR) or the outage probability as the performance criteria, different power allocation (PA) strategies are developed to exploit the knowledge of channel means. In~\cite{host2005capacity} bounds on the channel capacity are derived for a similar model with Rayleigh fading and channel state information (CSI). The power allocation problem in a three-node Gaussian orthogonal relay system is investigated in~\cite{liang2005gaussian} to maximize a lower bound on the capacity. Two power allocation schemes based on minimization of the outage probability are presented in~\cite{zhao2006improving} for the case when the information of the wireless channel responses or statistics is available at the transmitter. In~\cite{nguyen2011power} studies optimal power allocation schemes in a multi-relay cooperating network employing amplify-and-forward protocol with multiple source-destination pairs. The work~\cite{sanguinetti2018deep} advocates the use of deep learning to perform max-min and max-prod power allocation in the downlink of Massive MIMO networks. In~\cite{van2020joint} the total downlink power consumption at the access points is minimized, considering both to transmit powers and hardware dissipation.

The most relevant works to the current paper are of E.~Bj\"ornson et al. In~\cite[sec.~7.1]{Bjornson_tb_17} the case of single-antenna user equipment (UE) is studied in detail, targeting UE spectral efficiency and using multi-criteria optimization approach and Pareto front analysis. In~\cite[p.~328]{bjornson2013optimal} multi-antenna UEs are considered, but they are supposed to get only one data channel (or stream). The difficulty of the multi-antenna UE case is that the channels between different antennas of one UE are often spatial correlated~\cite{SpatialCorrelation}.
Therefore, the matrix of the user channel is ill-conditioned (or even has incomplete rank) thus one can not efficiently transmit data using the maximum number of streams. To solve this problem, instead of the full matrix of the user channel, vectors from its singular value decomposition (SVD) with the largest singular values are used for precoding~\cite{SVD}. 
When the number of streams (UE rank) is greater than one, it is necessary to consider the phenomenon of effective Signal-to-Interference-and-Noise-Ratio (effective SINR)~\cite{SINR_eff_model}. In~\cite{mohajer2022heterogeneous} a dynamic optimization model which maximizes the total energy efficiency along with satisfying the necessary QoS constraints is proposed. In~\cite{nikjoo2018novel} a novel approach to joint optimal power allocation and user association techniques in which cells are powered via a common grid network and alternative energy resources is suggested. In~\cite{mohajer2022energy} a dynamic optimization model to minimize the overall energy consumption of 5G heterogeneous networks is proposed.

In this paper, we study the problem of power allocation (PA) of MIMO wireless systems with users with multiple antennas and generalize the results of E.~Bj\"ornson et al. for the case of multi-antenna UEs with rank greater than one. We present the novel solutions to the PA problem that maximize network throughput in terms of spectral efficiency (SE) subject to either total or per-antenna power constraints. The original problem is not convex, but we managed to simplify it to a convex one with additional assumptions on the system model, e.g., applying a specific class of detection. Under some natural assumptions, we simplify the spectral efficiency functional and prove that the uniform power allocation provides its optimum subject to TPC. For the case of PAPC, we equivalently reformulate the optimization problem as the Lagrange system of equations and write down the Karush--Kuhn--Tucker conditions. Here, algorithmic solutions of PA problem are proposed assuming realistic PAPC.  

The simulation results based on Quadriga channel simulator~\cite{Quadriga} show the effectiveness of the proposed algorithmic approach in comparison with the reference PA schemes. To the authors’ best knowledge, these mathematical results are new.

The rest of this paper is organized as follows. After this Introduction, Section~\ref{sec:system_model} is devoted to the channel and system model where we introduce the downlink MIMO channel model, reference precoding methods, various detection schemes, and quality measures. In Section~\ref{sec:solution} we show a simplification of the PA problem, where we describe asymptotic diagonalization property of precoding matrices are used, proof of similarity of Conjugate and MMSE-IRC matrices, and Effective SINR models. In Section~\ref{sec:problem_solving} we consider the problem of the PA algorithm under TPC and PAPC assumptions, where we describe equal power allocation under the TPC, and the solution under the PAPC assumptions. We also consider problem-solving taking into account Modulation and Coding Scheme (MCS)~(\ref{sec:qam_solving}). The numerical algorithm description is presented in Section~(\ref{sec:algorithm}). Numerical experiments to compare considered algorithms are provided in Section~\ref{sec:results}. Algebraic notations and reference values are shown in Tab.~\ref{table_example}.

\begin{figure}
  \centering
  \includegraphics[width=0.7\linewidth]{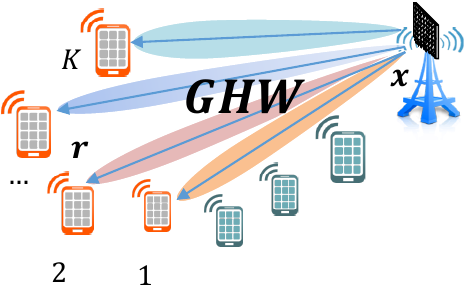} 
  \caption{Multi-User precoding allows transmitting of different information to different users simultaneously. Using the matrix $\bm W$ we can configure the amplitude and phase of the beams presented on the picture. The problem is to find the optimal precoding matrix $\bm W$ of the system given the target SE function~\eqref{Spectral Efficiency}.}
  \label{fig:system_example}
\end{figure}

\begin{table}
\small
\renewcommand{\arraystretch}{1.3}
\caption{Algebraic notations together with the reference values.}
\label{table_example}
\centering
\begin{tabular}{c c}
\hline
\bf Symbols & \bf Notations\\
\hline
$\bm H \in \mathbb{C}^{R \times T}, \bm W \in \mathbb{C}^{T \times L}, \bm G \in \mathbb{C}^{L \times R}$ & Channel, precoding and detection matrices\\
$\bm w_{n} \in \mathbb{C}^{T}$ & $n$-th column of matrix $\bm W$\\
$\bm h_{k} \in \mathbb{C}^T, \bm w^{k} \in \mathbb{C}^L$ & $k$-th row of matrices $\bm H, \bm W$\\
$h_{nm} \in \mathbb{C}, w_{nm} \in \mathbb{C}$ & $n,m$-th element of matrices $\bm H, \bm W$\\
$\bm S = \diag(s_1, \dots ,s_L) \in \mathbb{C}^{L \times L}$ & Diagonal matrix of singular values  \\
$K \; (=4)$ & the number of users\\
$T \; (=64)$ & the number of transmit antennas\\
$R \; (=16)$ & the total number of receive antennas\\
$R_k \; (=4)$ & the number of receive antennas for each user \\
$L \; (=8)$ & the total number of layers in the system\\
$L_k \; (=2)$ & the number of layers for each user \\
${\bf ()}^ \mathrm H$ & Complex conjugate operator\\
\hline
\end{tabular}
\end{table}


\section{Channel and System Model}\label{sec:system_model}

According to~\cite{Tse_tb_05,Bjornson_tb_17,Aitken,Zaidi_tb_18} we consider a MIMO broadcast channel. Symbol $\bm r \in \mathbb{C}^L $ is a \textit{received vector}, and $\bm x \in \mathbb{C}^L$ is a \textit{sent vector}, and $\bm H \in \mathbb{C}^{R \times T}$ is a \textit{channel matrix}, and $\bm W \in \mathbb{C}^{T \times L}$ is a \textit{Precoding matrix}, and $\bm G \in \mathbb{C}^{L \times R}$ is a block-diagonal \textit{detection matrix}, $\bm n \sim \mathcal{CN}(0, \sigma^2 I_{R}) $ is a \textit{noise-vector}, $\bm x \sim \mathcal{CN}(0, I_{L}) $ is a vector of sending symbols. Note that the linear precoding and detection are implemented by simple matrix multiplications. The constant $T$ is the number of transmit antennas, $R$ is the total number of receive antennas, and $L$ is the total number of transmitted symbols in the system. Usually, they are related as $L \leqslant R \leqslant T $. Each of the matrices $\bm G, \bm H, \bm W$ decomposes by $K$ \textit{users}, so please see the scheme in Fig.~\ref{fig:system_model}. The Multi-User MIMO model is described using the following linear system:
\begin{equation}
    \bm r = \bm G ( \bm H \bm W \bm x + \bm n ).
    \label{Basic received vector}
\end{equation}


In this paper, we make the following assumptions: i) that all users' channels are subject to uncorrelated Rayleigh fading, and ii) that the transmitter has perfect CSI of all downlink channels. This assumption is reasonable in time division duplex (TDD) systems because it enables the transmitter to use reciprocity to estimate the downlink channels. iii) that each user only has access to their own CSI. 

\begin{figure}
  \centering
  \includegraphics[width=\linewidth]{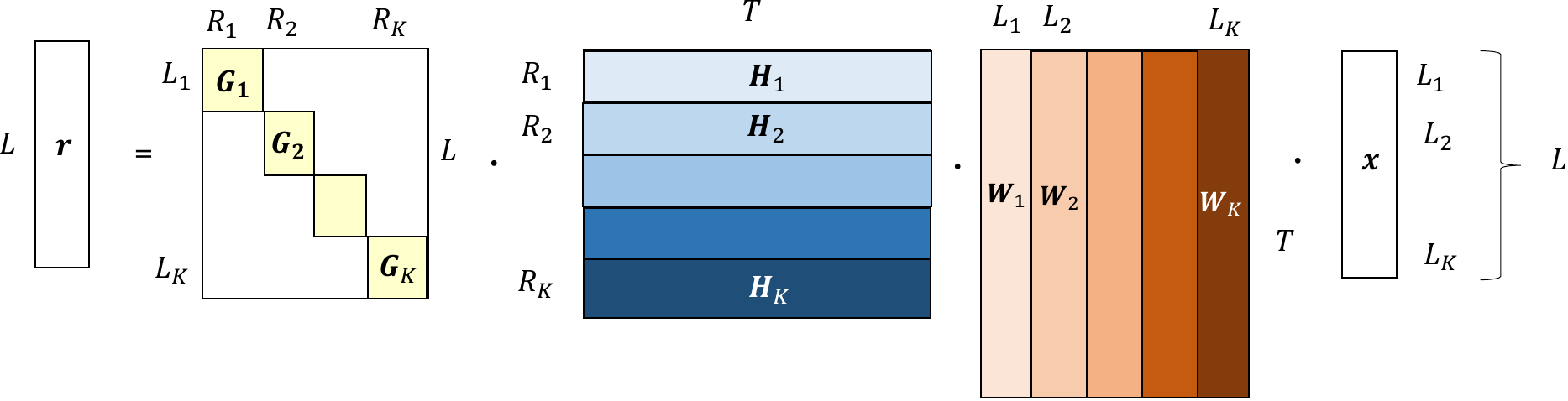}
  \caption{The linear system model  assumes only a linear transformation of the transmitting symbols.}
  \label{fig:system_model}
\end{figure}

\subsection{Singular Value Decomposition of the Channel}
The \textit{channel matrix} for user $k$, $\bm H_k \in \mathbb{C}^{R_k \times T}$ contains channel vectors $\bm h_i \in \mathbb{C}^{T}$ by rows. The \textit{path loss} diagonal matrix $\bm S_k \in \mathbb{R}^{R_k \times R_k}$ contains $R_k$ singular values $\sigma_{kn}$ in decreasing order along its main diagonal. It is convenient~\cite{SVD} to represent $\bm H_k $ via its Singular Value Decomposition (SVD): $\bm H_k = \bm U^\mathrm H_k \bm S_k \bm V_k$.



\begin{lemma}[Main Decomposition]\label{Main Decomposition}
\cite{Conjugate}~Denote $\bm H = [\bm H_1,\dots,\bm H_K] \in \mathbb C^{R\times T}$ the concatenation of individual channel rows $\bm H_k$. Similarly, $\bm U = \bdiag \{\bm U_1, \dots, \bm U_K\}$, $\bm S = \diag \{\bm S_1, \dots, \bm S_K\}$, $\bm V = [\bm V_1, \dots, \bm V_K]$.
Then, the decomposition exists (see Fig.~\ref{fig:main_decomposition}): $\bm H = \bm U^ \mathrm H \bm S \bm V$, where the $\bm H \in \mathbb{C}^{R \times T}$, 
and $\bm S = \diag(\bm S_k)\in \mathbb{C}^{R \times R}$,
and $\bm U = \bdiag (\bm U_k) \in \mathbb{C}^{R \times R}$ is \textit{block-diagonal unitary matrix},\ $\bm V = [\bm V_1,\dots,\bm V_K] \in \mathbb{C}^{R \times T}$ is the concatenation of corresponding UE singular vectors and $\bm C = \bm V \bm V^ \mathrm H - \bm I \neq \bm O$.
\end{lemma}
Lemma~\ref{Main Decomposition} means that by collecting all users together, we can write a specific \textit{channel matrix} decomposition~\cite{Conjugate}. Note, that such decomposition is not a convenient SVD of the channel matrix $\bm H$, and the matrix $\bm V$ is not unitary. But it consists of the $K$ SVDs of the size $R_k \times T$ and has block-diagonal unitary left matrix $\bm U$. We use this form in the construction of the optimal \textit{detection matrix} $\bm G$~\cite{SVD}.

\begin{figure}
  \centering
  \includegraphics[width=\linewidth]{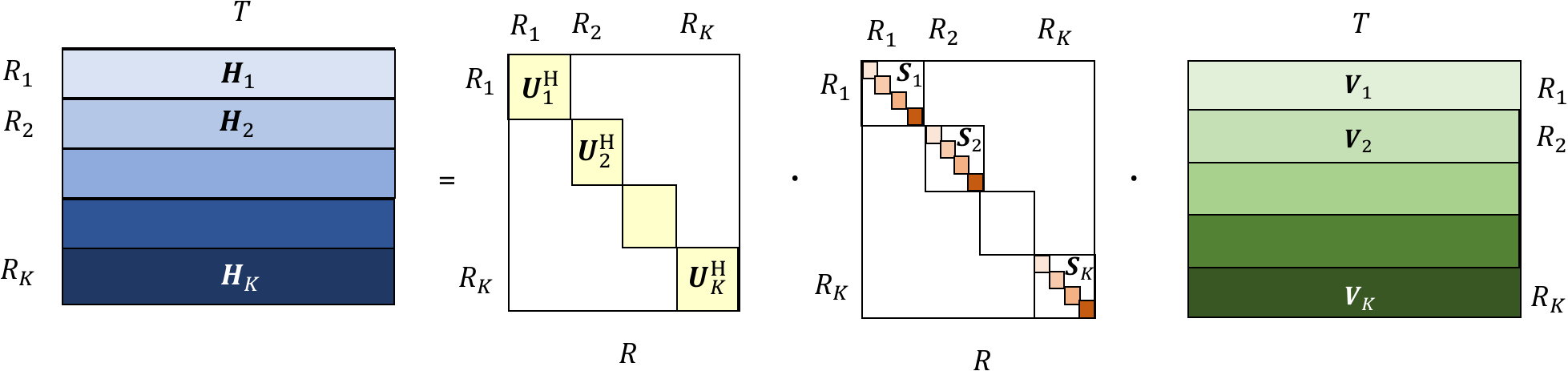}
  \caption{A graphical illustration of the Main Decomposition Lemma~\ref{Main Decomposition}.}
  \label{fig:main_decomposition}
\end{figure}

Usually, the transmitter sends to UE several layers and the number of layers (rank) is less than the number of UE antennas ($L_k \leqslant R_k$). In this case, it is natural to choose for transmission the first $L_k$ vectors from $\widetilde{\bm V}_k$ that correspond to the $L_k$ largest singular values from $\widetilde{\bm S}_k$. 
Denote by $\widetilde{\bm S}_k \in \mathbb{C}^{L_k \times L_k}$ the first $L_k$ largest singular values from $\bm S_k$, and by  $\widetilde{\bm U}_k^ \mathrm H \in \mathbb C^{R_k \times L_k}, \;  \widetilde{\bm V}_k \in \mathbb C^{L_k \times T}$ the first $L_k$ left and right singular vectors that correspond to $\widetilde{\bm S}_k$:
\begin{equation}
\widetilde{\bm S}_k = \diag\{s_{k,1}, \dots, s_{k,L_k}\}, \quad
\widetilde{\bm U}_k^ \mathrm H = (\bm u_{k,1}^ \mathrm H, \dots, \bm u_{k,L_k}^ \mathrm H), \quad
\widetilde{\bm V}_k = [\bm v_{k,1}; \dots; \bm v_{k,L_k}],
\end{equation}
i.e. ${\rm rank}  \widetilde{\bm V}_k = L_k \leqslant R_k = {\rm rank} \bm V_k$.
Numbers $L_k$ (and particular selection of $ \widetilde{\bm V}_k$) are defined during the Rank Adaptation problem that, along with Scheduler, is solved before precoding. 
For the Rank adaptation problem, we refer for example to~\cite{RankSelection} and in what follows we consider $L_k$, $ \widetilde{\bm V}_k$ already chosen.

\subsection{Precoding Matrices}\label{sec:precoding_class}

The \textit{precoding} matrix $\bm W$ is responsible for the beamforming from the base station to the users~\cite{VAE}. The linear methods for precoding do the following. Firstly, the linear solutions obtain singular value decomposition for each user $\bm H_k = \bm U_k^ \mathrm H \bm S_k \bm V_k \in \mathbb{C}^{R_k \times T}$ (Lemma~\ref{Main Decomposition}) and take the first $L_k$ singular vectors $ \widetilde{\bm V}_k \in \mathbb{C}^{L_k \times T}$ which attend to the first $L_k$ greatest singular values~\cite{SVD}. All these matrices are concatenated to the one matrix $\widetilde{\bm V} \in \mathbb{C}^{L \times T}$, which is used as the main building block of these precoding constructions. Finally, the precoding matrix is constructed from the obtained singular vectors. We describe linear methods for constructing a precoding matrix.

We are considering precoding matrices in the following form:
\begin{equation}\label{eq:precoding_power}
    \bm W=\bm W'\bm P, \qquad \bm W'=\bm W'(\widetilde{\bm V}),
\end{equation}
where $\widetilde{\bm V}$ is taken from the specific SVD decomposition from Lemma~\ref{Main Decomposition} and $\bm P$ is a diagonal matrix of power allocation.


Let us repeat some known precodings that are considered as initial solutions for studied power allocation methods.

The inter-user interference is vanished by the Zero-Forcing (ZF) precoding~\cite{ZF_MRT}:
\begin{equation}\label{Zero-Forcing}
    \bm W_{ZF} =  \widetilde{\bm V}^\dagger \bm P, \quad \widetilde{\bm V}^\dagger := \widetilde{\bm V}^ \mathrm H (\widetilde{\bm V} \widetilde{\bm V}^ \mathrm H)^{-1}
\end{equation}

It can be improved by using Regularized Zero-Forcing (RZF) precoding: 
\begin{equation}\label{RZF}
\bm W_{RZF}  =  \widetilde{\bm V}^ \mathrm H(\widetilde{\bm V} \widetilde{\bm V}^ \mathrm H +  \lambda \bm I )^{-1}\bm P,
\end{equation}
where the regularization parameter $\lambda = \frac{\sigma^2 L}{P} > 0$ depends on noise level and average path-losses~\cite{Joham_RZF}.

Further improvement is possible with diagonal regularization as in Adaptive Regularized Zero-Forcing (ARZF)~\cite{Conjugate} precoding (this idea was discussed in \cite{Bjornson_tb_17, nguyen2014mmse}, the following explicit heuristic formula for the MU MIMO case was proposed and studied in~\cite{Conjugate}):
\begin{equation}\label{ARZF}
\bm W_{ARZF}= \widetilde{\bm V}^ \mathrm H(\widetilde{\bm V} \widetilde{\bm V}^ \mathrm H + \lambda \bm S^{-2})^{-1} \bm P
\end{equation}
Detailed comparison of these algorithms and bibliography can be found in ~\cite{Conjugate}.

\subsection{Detection Matrices}\label{sec:conj_detection}

After precoding and transmission, on the side of UE $k$, we have to choose a  detection matrix $\bm G_k\in \mathbb C^{L_k\times R_k}$, which takes into account the rank of UE $L_k$. The way the UE performs detection strongly affects overall performance, and different detection algorithms require different optimal precoding matrices (see~\cite{Joham_RZF}, where precoding is chosen as a function of the detection matrix). The best way would be to consistently choose precoding and detection, but this is hardly possible due to the distributed nature of wireless communication. However, there are ideas on how to set up a precoding matrix, assuming a specific detection method on the UE side in the transmitter~\cite{PrecodingDetection}. We do not consider such an approach in our work, although it can be used to further improve our main proposal.

We assume the \textit{effective channel} matrix $ \bm A_k = \bm H_k \bm W_k $ to be calculated on the UE side. The Minimum Mean Square Error (\textit{MMSE}) detection for the user $ k $, where $\lambda \geqslant 0$ is the regularization value~\cite{MMSE, MMSE2}, performs as follows:
\begin{equation}\label{def:MMSE}
\bm G_k^{\textit{MMSE}}(\lambda) = (\bm A_k^ \mathrm H \bm A_k + \lambda \bm I) ^ {-1} \bm A_k ^ \mathrm H 
\end{equation}
In this paper, priority is given to the \textit{MMSE}-Interference-Rejection-Combiner (\textit{MMSE-IRC}) detection~\cite{IRC}:
\begin{equation}\label{def:IRC}
\bm G_k^{\textit{IRC}}(\lambda)  = \bm A_k ^ \mathrm H (\bm A_k^ \mathrm H \bm A_k + \bm R_{uu}^k + \lambda \bm I) ^ {-1}.
\end{equation}
And covariance matrix $\bm R_{uu}^k$ of total intra-user interference:
\begin{equation}
    \bm R_{uu}^k=\bm H_k(\bm W \bm W^ \mathrm H-\bm W_k \bm W_k^ \mathrm H)\bm H_k^ \mathrm H.
\end{equation}
To conduct analytical calculations, we assume virtual Conjugate Detection (\textit{CD}) in the following form~\cite{Conjugate}:
\begin{equation}\label{Conjugate Detection}
        \bm G^C_k = {\bm P_k}^{-1}\widetilde{\bm S}_k^{-1}  \widetilde{\bm U}_k = \bm P_k^{-1}  \widehat{\bm G}^C_k \in \mathbb C ^ {L_k \times R_k},
\end{equation}
where $\bm P_k$ is a corresponding to $k$-th user sub-matrix of matrix $\bm P$ in equation~\eqref{eq:precoding_power}.

\subsection{Quality Measures}\label{sec_QM}


We measure the quality of precoding using well-known functions such as Signal-to-Interference-and-Noise-Ratio (\textrm{SINR})~\cite{SINR} and Spectral Efficiency (\textrm{SE})~\cite{SE}. These functions are based not on the actual sending symbols $\bm x \in \mathbb C ^ {L\times 1}$, but some distribution of them~\cite{Bjornson}. Thus, we get the common function for all assumed symbols, which can be sent using the specified precoding matrix. We denote $\mathcal{L}_k$ as the set of symbols for $k$-th user. The \textrm{SINR} function is defined as:
\begin{equation}\label{Symbol SINR}
\textrm{SINR}_l(\bm W, \bm H_k, \bm g_l, \sigma^2) := \dfrac{|\bm g_l \bm H_k \bm w_l |^2}{\sum_{i \ne l}^{L} | \bm g_l \bm H_k \bm w_i |^2 + \sigma^2 \|\bm g_l\|^2}, \quad \forall l \in \mathcal{L}_k.
\end{equation}
For simulations of a physical layer (PHY) in multi-carrier and multi-layer OFDM systems an effective SINR mapping (ESM) model is used. Such model compresses the given set of SINRs experienced by the receiver over every sub-channel into a single scalar value (called effective SINR).
According to the paper~\cite{SINR_eff_model}, the \textit{effective} SINR for a user $k$ is calculated using the \textrm{SINR} at each layer of each Resource Block (RB) as follows. Functions $\beta = \beta(\textrm{MCS})$ and $\textrm{MCS} = \textrm{MCS}(\textrm{SINR}^{eff}_\beta)$ are table-defined (see, e.g., Table~\ref{tab:beta} for $\beta(MCS)$).  Assuming only one RB, we can define $\textrm{SINR}^{eff}_\beta$ as a self-consistent solution of the following system:
\begin {equation}
\textrm{SINR}_{\beta, k} ^  {eff} (\bm W, \bm H_k, \bm G_k, \sigma^2) = - \beta \ln
\left (
\frac {1} {L_k} \sum \limits_{l \in \mathcal{L}_k} \exp  \Big\{- \frac {\textrm{SINR}_ {l} (\bm W, \bm H_k, \bm g_l, \sigma^2) } {\beta}\Big\}
\right)
\label{SINR_eff}
%
\end {equation}

This model is called Exponential effective SINR mapping (EESM)
and the accuracy of EESM has been validated in several
studies~\cite{QAM64, brueninghaus2005link, 6678684}.
To get the \textrm{SE} function, we apply Shannon’s formula over all effective user SINRs~\eqref{User SINR}:
\begin{equation}
    \textrm{SE}(\bm W, \bm H, \bm G, \sigma^2) =  \sum_{k=1}^K L_k \log_2 (1 + \textrm{SINR}_{\beta, k}^{eff}(\bm W, \bm H_k, \bm G_k, \sigma^2)) \rightarrow \max\limits_{\bm W}.
    \label{Spectral Efficiency}
\end{equation}

\subsection{Problem Statement}

We consider the channel model in the form~\eqref{Basic received vector} that particularly means exact measurements of the channel. To further simplify the problem we suppose detection policy $\bm G = \bm G(\bm H, \bm W)$ to be a known function, moreover we assume Conjugate Detection~\eqref{Conjugate Detection} that simplifies the channel model to~\eqref{model_conj}. Based on this channel model, we calculate \textrm{SINR} of transmitted symbols by~\eqref{Symbol SINR} and effective \textrm{SINR} of UE, which can be approximately calculated by~\eqref{SINR_eff} and~\eqref{User SINR}. We denote the total power of the system as, $P$, \textcolor{black}{ assuming $P=1$ in the experiments.} 

The \textit{total power constraint} and the more realistic \textit{per-antenna power constraints} (see~\cite{Bjornson_tb_17}) impose the following conditions on the precoding matrix. Since case $\bm W=\bm W'(\bm V) \bm P$ is considered in this paper, conditions read:
\begin{equation}\label{power_antenna}
{\rm (a)} \;\; \|\bm W' \bm P\|^2 \leqslant P, \qquad \text{or} \qquad 
{\rm (b)} \;\; \|{\bm w'}^{t}{\bm p}\|^2 \leqslant P/T, \quad t=1, \dots ,T,
\end{equation}

where $\bm P=\diag(\bm p)= \diag\left(\sqrt{p_1} \dots \sqrt{p_L}\right) = \diag\left(\frac{\sqrt{\rho_1}}{\|{\bm w'}_1\|} \dots \frac{\sqrt{\rho_L}}{\|{\bm w'}_L\|}\right)$ is power allocation matrix and $P$ is total power of base station. {\bf The  goal is to find a \textcolor{black}{power allocation} matrix that maximizes \textrm{SE}~\eqref{Spectral Efficiency} given the power constraints}~\eqref{power_antenna}:
\begin{equation}\label{First formulation}
	\textrm{SE}(\bm P) = \textrm{SE}(\bm W' \bm P, \bm H, \bm G(\bm H, \bm W'\bm P), \sigma^2)   \rightarrow \max\limits_{\bm P}, \quad\text{subject to  (a) or (b).}
\end{equation}



\section{Simplifications of the Problem}\label{sec:solution}

\subsection{Asymptotic Diagonalization Property of Precoding}

\begin{definition}
Let us assume the case of small noise and denote $\lambda=\frac{\sigma^2}{P} \to 0$ and $\bm P > 0$ is some diagonal matrix. In real systems, Scheduler algorithms choose UE for pairing if this assumption is fulfilled. Define the property of {\bf asymptotic diagonalization of $\widetilde {\bm V}$ as $\lambda\to0$} of precoding matrix as follows:
\begin{equation}\label{IRC 1}
\widetilde{\bm V} \bm W = 
\left( \begin{array}{c}
	\widetilde{\bm V}_{1}\\
	\widetilde{\bm V}_{2}\\
	 \ldots \\
	\widetilde{\bm V}_{K}\\
	\end{array} \right)  \cdot  \left(\bm W_{1}, \bm W_{2} \ldots \bm W_{K} \right) =\bm P + \mathcal{O}(\lambda), \ \text{i.e.} \ \widetilde{\bm V} \bm W\sim \bm P, \text{ as } \lambda \rightarrow 0
\end{equation}
\end{definition}

Precoding algorithms: ZF~\eqref{Zero-Forcing}, RZF~\eqref{RZF}, and ARZF~\eqref{ARZF} satisfy the property~\eqref{IRC 1}. This can be easily shown with the Neumann series as in the following Lemma (it is similar to~\cite[Lemma~2]{Conjugate}).

\begin {lemma} \label{lemm inv}
Consider square invertible complex matrices $ \bm M $ and $ \bm N $ of the same size and rank. For any $ 0 <\lambda \ll 1 $ and $\det \bm M \neq 0$ the following matrix identity is true: $(\bm M + \lambda \bm N) ^ {- 1} = \bm M^{-1} - \lambda \bm M^{-1} \bm N \bm M^{-1} + \mathcal{O}(\lambda^2) = \bm M ^ {- 1} + \mathcal{O}(\lambda)$.
\begin{proof}

\begin{equation}
\bm F(\lambda) = (\bm M + \lambda \bm N)^{-1}, \text{ and } \bm F'(\lambda) = -(\bm M + \lambda \bm N)^{-1} \bm N (\bm M + \lambda \bm N)^{-1}
\end{equation}
\begin{equation}
\bm F(\lambda) = \bm F(0) + \bm F'(0)\lambda + \mathcal{O}(\lambda^2), \text{ where } \bm F(0) = \bm M^{-1}, \text{ and } \bm F'(0) = -\bm M^{-1} \bm N \bm M^{-1}
\end{equation}
\begin{equation}
(\bm M + \lambda \bm N) ^ {- 1} = \bm M^{-1} - \lambda \bm M^{-1} \bm N \bm M^{-1} + \mathcal{O}(\lambda^2) = \bm M ^ {- 1} + \mathcal{O}(\lambda)
\end{equation}
\end{proof}
\label{epslem}
\end {lemma}

For channel singular values $\widetilde{\bm V}$ such that the matrix $\widetilde{\bm V} \widetilde{\bm V}^ \mathrm H$ has a full rank, using Lemma~\ref{lemm inv} for the algorithms ZF~\eqref{Zero-Forcing}, RZF~\eqref{RZF} and ARZF~\eqref{ARZF} we obtain:
\begin{equation}
 \widetilde{\bm V}{\bm W}_{ZF} =  \widetilde{\bm V}\widetilde{\bm V}^ \mathrm H ( \widetilde{\bm V}  \widetilde{\bm V}^ \mathrm H)^{-1} \bm P = {\bm P}
\end{equation}
\begin{equation}
 \widetilde{\bm V}{\bm W'}_{RZF}
 = \widetilde{\bm V}\widetilde{\bm V}^ \mathrm H ( \widetilde{\bm V}  \widetilde{\bm V}^ \mathrm H + \lambda \bm I)^{-1} \bm P
 = {\bm P} + \mathcal{O}\left(\lambda\right) 
\end{equation}
\begin{equation}
 \widetilde{\bm V}{\bm W'}_{ARZF}
 = \widetilde{\bm V}\widetilde{\bm V}^ \mathrm H ( \widetilde{\bm V}  \widetilde{\bm V}^ \mathrm H + \lambda \bm S)^{-1} \bm P
 = {\bm P} + \mathcal{O}\left(\lambda\right) 
\end{equation}

Thus, precodings ZF~\eqref{Zero-Forcing}, RZF~\eqref{RZF}, and ARZF~\eqref{ARZF} satisfy property~\eqref{IRC 1}.

\begin{remark}
In this case, matrix $\bm P$ of definition~\eqref{IRC 1} coincides with matrix $\bm P$ of Conjugate Detection~\eqref{Conjugate Detection}. 
\end{remark}

\subsection{The Similarity of Conjugate Detection and MMSE-IRC}\label{sec:similarity_proof}

In this section, we prove the similarity of \textit{MMSE-IRC}~\eqref{def:IRC}~\cite{IRC} and \textit{Conjugate Detection} (\textit{CD})~\eqref{Conjugate Detection}~\cite{Conjugate}. Detection \textit{CD} does not depend on precoding and allows to significantly simplify the considered problem~\eqref{First formulation}. First, we prove some useful properties about \textit{CD} (compare with~\cite[Theorem~1]{Conjugate}).

\begin{lemma}\label{lemma_conj_det}
The detection matrix $\bm G$ is $\bm G^C$ (\textit{Conjugate Detection}) if and only if it satisfies the following property: 
\begin{equation}\label{IRC 4}
\bm G = \bm G^C \Leftrightarrow \bm G \bm H = {\bm P}^{-1} \widetilde{\bm V} \Leftrightarrow \forall k:
\bm G_{k}\bm H_{k}=\bm P_k^{-1}\widetilde{\bm V}_{k},
\end{equation}
where $\bm P$ is uniquely defined in~\eqref{IRC 1}, and the system model equation~\eqref{Basic received vector} takes the form
\begin{equation}\label{model_conj}
\bm r = \widetilde{\bm V} \bm W \bm x + \tilde {\bm n}, \quad 
\tilde {\bm n} := {\bm P}^{-1} \widetilde{\bm S}^{-1} \widetilde{\bm U} \bm n.
\end{equation}

\end{lemma}
\begin{proof}
Necessity. Using Lemma~(\ref{Main Decomposition}) we can write 
\begin{equation}
    \bm G^C \bm H = {\bm P}^{-1}\widetilde{\bm S}^{-1} \widetilde{\bm U} \bm U^ \mathrm H \bm S \bm V = {\bm P}^{-1}\widetilde{\bm S}^{-1} \big[\begin{array}{c|c} \bm I & \bm O \end{array}\big] \bm S \bm V =  {\bm P}^{-1}\widetilde{\bm S}^{-1} \widetilde{\bm S} \widetilde{\bm V} = {\bm P}^{-1}\widetilde{\bm V},
\end{equation}

which immediately leads to~\eqref{model_conj}.

Sufficiency. Assume that~\eqref{IRC 4} holds, then $\widetilde{\bm V} = \bm P \bm G \bm H  $, since the matrix $\bm P > O$. Then, $\forall \bm v \in \widetilde{\bm V}$ expansion of vector~$\bm v$ in basis $\bm H$ is unique. The elements of the matrix $\bm P \bm G$ are the coefficients of this expansion. Therefore, a matrix ${\bm G}$ with the property~\eqref{IRC 4} is unique. 

The last equivalence in~\eqref{IRC 4} is true due to the block diagonality of the matrix $\bm G$.
\end{proof}

\begin{theorem}
In assumption that $\bm H_k$ has the full rank and precoding $\bm W$ has property~\eqref{IRC 1}, detection $\bm G^{\textit{IRC}} (\lambda) $~\eqref{def:IRC} asymptotically equals to $\bm G^{C}$~\eqref{Conjugate Detection}, in other words $\bm G^{\textit{IRC}}(\lambda) \sim  \bm G^C  \text{ as } \lambda\rightarrow 0$.
\end{theorem}

\begin{proof}
We need the following consequence of the (\ref{IRC 1}) property:
\begin{equation}\label{IRC 2}
\bm W \bm W^ \mathrm H\widetilde{\bm V}_{k}^ \mathrm H= \left(  \sum _{v=1}^{K}\bm W_v\bm W_v^ \mathrm H \right) \widetilde{\bm V}_{k}^ \mathrm H \sim \bm W_{k}\bm W_{k}^ \mathrm H\widetilde{\bm V}_{k}^ \mathrm H \sim \bm W_{k} \bm P_k
\end{equation}
Taking into account the form of \( \bm R_{uu}^{k}  \) we can rewrite~\cite{QNS}:
\begin{multline*}
    \bm G^{\textit{IRC}}_k(\lambda) = \left( \bm H_{k}\bm W_{k} \right) ^ \mathrm H (\bm H_k \bm W_k  (\bm H_k \bm W_k)^ \mathrm H + \bm R_{uu}^{k}  + \lambda \bm I)^{-1} = \\ = \left( \bm H_{k}\bm W_{k} \right) ^ \mathrm H (\bm H_k \bm W_k  (\bm H_k \bm W_k)^ \mathrm H +  \bm H_k(\bm W \bm W^ \mathrm H-\bm W_k \bm W_k^ \mathrm H)\bm H_k^ \mathrm H  + \lambda \bm I)^{-1} = \\ = \left( \bm H_{k}\bm W_{k} \right) ^ \mathrm H (\bm H_k \bm W (\bm H_k \bm W)^ \mathrm H + \lambda \bm I)^{-1}.
\end{multline*}

Using (\ref{IRC 4}), (\ref{IRC 2}),  Lemma~\ref{lemma_conj_det} in the case $\lambda\rightarrow 0$ we obtain:
\begin{multline*}
    \bm G_k^C= \bm I \bm G_k^C \bm I = \bm P_k^{-1} \bm P_k \bm G_k^C ({\bm H}_{k} \bm W (\bm H_k \bm W)^ \mathrm H) (\bm H_k \bm W (\bm H_k \bm W)^ \mathrm H)^{-1}=
    \{Eq.~\ref{IRC 4}\} = \\
    =\bm P_k^{-1} \widetilde{\bm V}_{k} \bm W \bm W ^ \mathrm H \bm H_{k} ^ \mathrm H (\bm H_k \bm W (\bm H_k \bm W)^ \mathrm H)^{-1} \sim \{Eq.~\ref{IRC 2}\} \sim \\
    \sim  \bm W_{k}^ \mathrm H \bm H_{k}^ \mathrm H (\bm H_k \bm W (\bm H_k \bm W)^ \mathrm H)^{-1} 
    = \left( \bm H_{k}\bm W_{k} \right) ^ \mathrm H (\bm H_k \bm W (\bm H_k \bm W)^ \mathrm H)^{-1}  \sim  \bm G^{\textit{IRC}}_k .
\end{multline*}

\end{proof}

\begin{remark}
The introduced \textit{CD} detection is speculative: it hardly can be implemented in practice. UE measures $\bm H_k \bm W_k$ via pilot signals instead of $\bm H_k$. Nonetheless, it is very useful for theoretical research.  Moreover, the asymptotic behavior of \textit{MMSE} and \textit{MMSE-IRC} detection is similar to that of \textit{CD} (Sec.~\ref{sec:similarity_proof}). Particularly, if precoding $\bm W$ is Zero-Forcing~\eqref{Zero-Forcing} and the noise power is zero ($\sigma^2 = 0$), then $\bm G^{\textit{IRC}} (\lambda) = \bm G^C$; if, additionally, precoding has the full rank, then $\bm G^{\textit{MMSE}} (\lambda) = \bm G^C$. 
\end{remark}

\begin{remark}\label{Reducing to Layers}
Lemma~\ref{lemma_conj_det} shows that the assumption that UEs use \textit{CD} on their side sufficiently simplifies the initial problem, decreases its dimensions, and allows notation to be uniform. Namely, we can work with \textit{user layers} of shapes $L_k$ and $L$ instead of considering \textit{user antennas} space.
Note also that for precoding it is sufficient to only perform Partial SVD of the channel $\bm H_k \in \mathbb C ^ {R_k \times T}$, keeping just the first $L_k$ singular values and vectors for each user $k$: $\bm H_k \approx \widetilde{\bm U}^ \mathrm H_k \widetilde{\bm S}_k \widetilde{\bm V}_k$.

Based on this, in what follows we can omit the tilde and write $\bm U_k, \bm S_k, \bm V_k$ instead of $\widetilde{\bm U}_k, \widetilde{\bm S}_k, \widetilde{\bm V}_k$ correspondingly.
\end{remark} 

\subsection{Low Correlated Users}\label{sec:low_cor}

We define an \textit{interference-correlation matrix} as $\bm C = \bm V \bm V^ \mathrm H - \bm I$. In real networks, the set of UEs is chosen by Scheduler and the number of layers of each UE is chosen to be fixed by the Rank Selection algorithm. Both Scheduler and Rank Selection methods provide $\|\bm C\| = \mathcal{O}(\lambda)$, where $\lambda=\frac{\sigma^2}{P}$ is the noise-power ratio. Thus, we assume user correlation to be low compared to noise power, which means $ \| \bm C \|  = \mathcal{O}(\lambda)$.

\begin{lemma}\label{lemma:singular_channel}
    For precoding $\bm W = \bm W' \bm P $ satisfying the property (\ref{IRC 1}) and inference-correlation matrix $\bm C = \bm V \bm V^ \mathrm H - \bm I$ satisfying $  \| \bm C \| = \mathcal{O}(\lambda)$, is the noise-power ratio, it is asymptotically true that $\bm G^C \bm H \bm W = (1 - \lambda)\bm I + \mathcal{O}(\lambda^2)$.
\end{lemma}
\begin{proof}
    \begin{multline}\label{eq:singular_channel}
        \bm V \bm W' = \bm V \bm V^ \mathrm H (\bm V \bm V^ \mathrm H + \lambda \bm I)^{-1} = \{Lemma~\ref{lemm inv} \}  =  \bm V \bm V^ \mathrm H ((\bm V \bm V^ \mathrm H)^{-1} - \lambda (\bm V \bm V^ \mathrm H)^{-2} + \mathcal{O}(\lambda^2))  = \\ = \bm I - \lambda (\bm V \bm V^ \mathrm H)^{-1} + \mathcal{O}(\lambda^2) = \bm I - \lambda (\bm C + \bm I)^{-1} + \mathcal{O}(\lambda^2) = \bm I - \lambda (\bm I + \mathcal{O}(\|\bm C\|)) + \mathcal{O} (\lambda^2) = \\ = (1 - \lambda) \bm I + \lambda \mathcal{O}(\|\bm C\|)  + \mathcal{O} (\lambda^2) = \{  (\| \bm C \| ) = \mathcal{O}(\lambda) \} = (1 - \lambda) \bm I + \mathcal{O}(\lambda^2)
    \end{multline}
    \begin{multline*}
        \bm G^C \bm H \bm W = \{Lemma~\ref{lemma_conj_det}\} = \bm P^{-1} \bm V \bm W = \bm P^{-1} \bm V \bm W' \bm P  = \{Eq.~\ref{eq:singular_channel} \} =  \\ = \bm P^{-1} (1 - \lambda)  \bm I \bm P  + \bm P^{-1} \mathcal{O} (\lambda^2) \bm P = (1 - \lambda)\bm I + \mathcal{O}(\lambda^2)
    \end{multline*}
\end{proof}

Using Lemma~\ref{lemma:singular_channel} we immediately get the following
\begin{theorem}
    For precoding $\bm W$ satisfying the property (\ref{IRC 1}) and inference-correlation matrix $\bm C = \bm V \bm V^ \mathrm H - \bm I$ satisfying $ \| \bm C \| = \mathcal{O}(\lambda)$, where $\lambda=\frac{\sigma^2}{P}$ is the noise-power ratio, formula for \textrm{SINR} (\ref{Symbol SINR}) in the case of $\bm G^C$~\eqref{Conjugate Detection} detection will take the asymptotic form:
    \end{theorem}
        \begin{equation}
        \label{SINR layer}
        \textrm{SINR}_l(\bm W, \bm H_k, \bm g_l^C, \sigma^2) \sim \frac{p_{l} s_{l}^{2}}{\sigma^2}
    \end{equation}
\begin{proof}
    \begin{multline*}
        \textrm{SINR}_l(\bm W, \bm H_k, \bm g_l^C, \sigma^2) := \dfrac{|\bm g^C_l \bm H_k \bm w_l |^2}{\sum_{i=1, \ne l}^{L} | \bm g^C_l \bm H_k \bm w_i |^2 + \sigma^2 \|\bm g^C_l\|^2}  = \{Lemma~\ref{lemma:singular_channel} \} = \\ = \dfrac{1 - \lambda + \mathcal{O}(\lambda^2)}{\mathcal{O}(\lambda^2) + \dfrac{\sigma^2}{p_{l} s_{l}^{2}}} =  \dfrac{1 - \dfrac{\sigma^2}{P} + \mathcal{O}\Big(\dfrac{\sigma^4}{P^2}\Big)}{\mathcal{O}\Big(\dfrac{\sigma^4}{P^2}\Big) + \dfrac{\sigma^2}{p_{l} s_{l}^{2}}}  \sim \dfrac{p_{l} s_{l}^{2}}{\sigma^2}
    \end{multline*}
\end{proof}


\subsection{Effective SINR Models}\label{sec:conjugate_sinr}

In this subsection, we compare two models of Effective SINR from~\cite{QAM64, SINR_eff_model, Conjugate}. In theoretical calculations, model (\ref{SINR_eff}) is extremely inconvenient. To simplify the formula of effective \textrm{SINR}~\eqref{SINR_eff}, we average $L_k$ per-symbol \textrm{SINR}s~\eqref{Symbol SINR} by the geometric mean, where $\mathcal{L}_k$ denotes the set of symbols for $k$-th user:
\begin{equation}
    \textrm{SINR}_k^{eff}(\bm W, \bm H_k, \bm G_k,  \sigma^2) = \Big({\prod\nolimits_{l \in \mathcal{L}_k} \textrm{SINR}_l(\bm W, \bm H_k, \bm g_l, \sigma^2) } \Big)^{\frac{1}{L_k}}, \quad \forall l \in \mathcal{L}_k.
    \label{User SINR}
\end{equation}

Fig.~\ref{fig:qam} shows the dependencies of $\textrm{SINR}^{eff} (dB)$ for a user with four antennas to justify the close relationship of the various \textrm{SINR} averaging~\eqref{SINR_eff} and~\eqref{User SINR}. The $ x $ axis is the 
average SINR in dB: $\frac{1}{4}\sum_{l=1}^4\textrm{SINR}_l(dB)$.

Fig.~\ref{fig:qam} shows the comparison of effective \textrm{SINR} in the form of the geometric mean and the form of different MCS-$\beta$ values. Differences between various effective \textrm{SINR}s can take values greater than five decibels. On the other hand, points \textrm{SINR} with a large difference in the maximum and minimum values are unusual in practice.

\begin{figure}
    \centering
    \includegraphics[width = \linewidth]{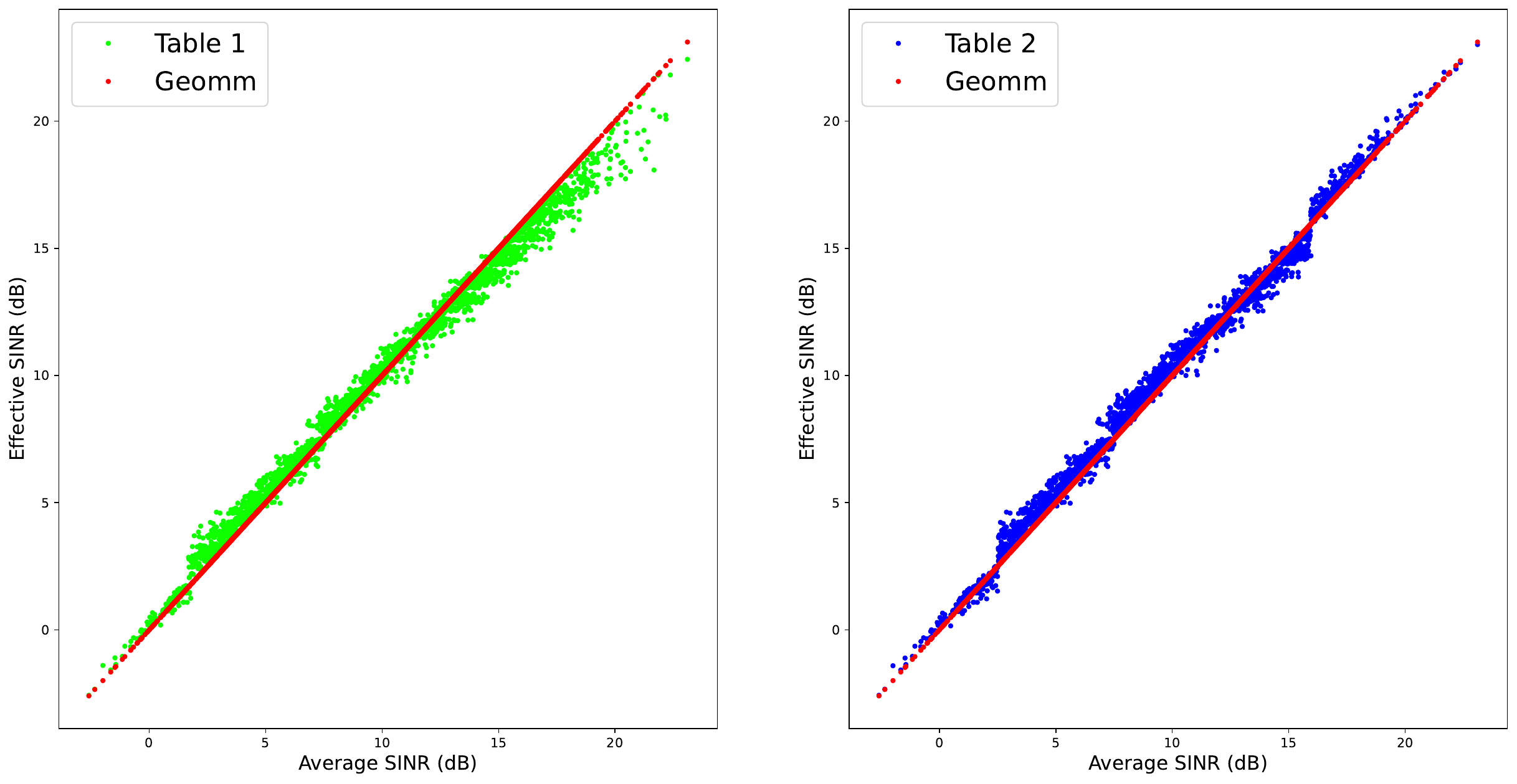}
    \caption{Approximation of exponential model of $\textrm{SINR}^{eff}$~\eqref{SINR_eff} realized with MCS Tables 1 and 2 (green and blue points). Geomm is an acronym of Geometrical Mean~\eqref{User SINR} (red points).}
    \label{fig:qam}
\end{figure}

For precoding $\bm W$ satisfying the property (\ref{IRC 1}) and from the formula for \textrm{SINR} (\ref{SINR layer}) with a \textit{CD} using the geometric mean effective \textrm{SINR} model (\ref{User SINR}), we can write the \textrm{SINR} for the $k$-th user as follows:
\begin{equation}
\textrm{SINR}^{eff}_k(\widetilde{\bm S}_k, \bm P_k, \sigma^2) = \frac{1}{\sigma ^ 2} \sqrt[L_k]{\prod\limits_{l = 1}^{L_k}(s_l ^ 2 p_l)}.
\label{SU SINR}
\end{equation}

The formula (\ref {SU SINR}) reflects the channel quality for the specified user without considering other users. The value of $\textrm{SINR}^{eff}_k(\widetilde{\bm S}_k, \bm P_k, \sigma^2) $ depends on the singular values $ \widetilde{\bm S}_k \in \mathbb{R} ^ {L_k \times L_k} $ (related to matrices $ \bm H_k \in \mathbb {C} ^ {R_k \times T} $), the transmitted power $ \bm P_k $ and noise $ \sigma ^ 2 $. This function will be used in theoretical calculations due to its simplicity.

\subsection{Spectral Efficiency Simplification}

In this section we simplify optimization problem of Spectral Efficiency maximization in case of Zero-Forcing algorithms with asymptotic diagonalization property~\eqref{IRC 1}, Conjugate detection matrix~\eqref{Conjugate Detection}, geometrical averaging of effective $\textrm{SINR}$ model~\eqref{SU SINR} and low correlated users (see Lemma~\ref{lemma:singular_channel}).

For any $x \gg 1$ it is true that: $ \log ( 1 + x ) = \log x + O ( x^{-1})$, and so
\begin{multline}
\label{SE_geom1}
    SE(\bm W, \bm V, \sigma^2) =  \sum_{k=1}^K L_k \log_2 (1 + \textrm{SINR}_k^{eff}(\bm W, \bm V_k, \bm S_k, \sigma^2)) =  \\=  \sum_{k=1}^K L_k \log_2 (\textrm{SINR}_k^{eff}(\bm W, \bm V_k, \bm S_k, \sigma^2)) + \sum_{k=1}^K \mathcal{O}(\textrm{SINR}_k^{eff (-1)}(\bm W, \bm V_k, \bm S_k, \sigma^2)) 
\end{multline}
We simplify the initial optimization problem by maximization of its leading term:
\begin{multline}
\label{SE_geom2}
    \sum_{k=1}^K L_k \log_2 (\textrm{SINR}_k^{eff}(\bm W, \bm V_k, \bm S_k, \sigma^2)) = \sum_{k=1}^K L_k \log_2 \Big({\prod\nolimits_{l \in \mathcal{L}_k} \textrm{SINR}_l(\bm W, \bm H_k, \bm g_l, \sigma^2, P) } \Big)^{\frac{1}{L_k}} = \\ = \sum_{k=1}^K  \log_2 {\prod_{l \in \mathcal{L}_k} \textrm{SINR}_l(\bm W, \bm H_k, \bm g_l, \sigma^2, P) }\rightarrow \max\limits_{\bm P}  
\end{multline}
These problems are not equivalent, although their solutions are close to each other. 
If we calculate $\bm W$ by ZF algorithm, that gives zero interference, then \textrm{SINR} is as follows
\begin{equation}
    \textrm{SINR}_{l}(\bm W, \bm v_{l}, s_{l}, \sigma^2) = \text{\{Zero-Forcing Algorithm\}} = \frac{s_{l}^2}{\sigma^2} p_{l}
    \label{sinr_zf}
\end{equation}
and maximization of the leading term gives
\begin{multline*}
    \sum_{k=1}^K \log_2 \prod_{l \in \mathcal{L}_k} \textrm{SINR}_{l}(\bm W, \bm v_{l}, s_{l}, \sigma^2) = \sum_{k=1}^K \log_2\prod_{l \in \mathcal{L}_k} \frac{s_{l}^2}{\sigma^2} p_{l} = \sum_{k=1}^K \log_2 \prod_{l \in \mathcal{L}_k} \frac{s_{l}^2}{\sigma^2}  \prod_{l \in \mathcal{L}_k} p_{l} = \\\sum_{j=1}^K  \log_2 \prod_{l \in \mathcal{L}_k} s_{l}^2 - \sum_{j=1}^K  \log_2 \prod_{l \in \mathcal{L}_k} \sigma^2 + \sum_{j=1}^K \log_2 \prod_{l \in \mathcal{L}_k} p_{l}  \rightarrow \max\limits_{\bm P}
\end{multline*}
Finally, we can reduce tasks~\eqref{First formulation} (a) and (b) to the following problems: 
\begin{equation}
    \sum_{k=1}^K \log_2 \prod_{l \in \mathcal{L}_k} p_{l} = \log_2 \prod_{l=1}^L p_l \rightarrow \max\limits_{\bm P}, \quad\text{s.t.     } 
    ||\bm W||^2 \leqslant P.
\end{equation}

\section{Solutions of the Problem}\label{sec:problem_solving}
According to~\cite[sec.~7]{Bjornson_tb_17} we consider equal transmit power strategy for all $K$ users. Such power allocation  gives the maximum for a reasonable lower bound on the $SE$ (\ref{Spectral Efficiency}) under some feasible assumptions. Although this Power allocation is not optimal, these heuristics provide a good suboptimal solution.

\subsection{Simplified PA Problem with Total Power Constraints}

\begin{theorem}
If $\bm W$ satisfies to the property~\eqref{IRC 1} and $\bm G = \bm G^{C}$, assuming model~\eqref{User SINR} of effective \textrm{SINR}, the equal PA (all $\|\bm w_l\|$ is equal, namely, $p_l=P/L$) asymptotically provides maximum to the first optimization problem:
\begin{equation}
U = \sum_n SE_n \rightarrow \max , \quad ||\bm W||^2 \leqslant P.
\end{equation}

\end{theorem}
\begin{proof}
Using asymptotic $\ln (1 + \textrm{SINR}) = \ln(\textrm{SINR})(1 + O(\e))$ for large \textrm{SINR}, conjugate detection, \textrm{SINR} estimation (\ref{sinr_zf}) for ZF algorithm and considering coordinates $\rho_l$ we get first optimization problem (\ref{First formulation}):
\begin{equation}
    \prod_{l=1}^L \rho_{l} \rightarrow \max\limits_{\rho_1 \dots \rho_L}, \quad\text{s.t.     } 
    \sum\limits_{l=1}^L || {\bm w'}_l||^2\frac{\rho_l}{\|{\bm w'}_l\|^2}=\sum\limits_{l=1}^L \rho_l \leqslant P.
\end{equation}
It is an optimization problem of the maximal volume of the box with predefined lengths of edges which solution is 
\begin{equation}
\label{eq_pa}
\forall l: \rho_l=P/L, \text{ and } p_l=\frac{P/L}{||{\bm w'}_l||^2}
\end{equation}
\end{proof}

\begin{remark}
The original function~\eqref{Spectral Efficiency},~\eqref{User SINR} asymptotically reaches its maximum at the solution of the simplified  PA problem~\eqref{eq_pa}.
\end{remark}

\subsection{EESM Model and Total Power Constraints}\label{sec:qam_solving}

By analogy with the formulas~(\ref{SE_geom1} and~\ref{SE_geom2}) we can calculate Spectral efficiency using physical MCS-$\beta$ model~\eqref{SINR_eff}, where the parameter $\beta_k$ for each $k=1 \dotsc K$ depends on given $MCS$ and therefore depends on the precoding matrix, in particular on power allocation variables $p_l$ for all $l = 1 \dots L$:
\begin{multline}\label{SE_calc_beta_fixed}
SE(\bm W, \bm H, \sigma^2)=\sum_{k=1}^K L_k \ln (1+SINR_k^{eff}) = \\ = -
\sum_{k=1}^K L_k\ln\left (1 - \beta_k \log \left (
\frac {1} {L_k} \sum \limits_ {j = 1} ^ {L_k} \exp \left(- \frac {SINR_{kj}} {\beta_k}\right)
\right)\right)
\end{multline}

The function~\eqref{SE_calc_beta_fixed} is discontinuous. Nevertheless, if we fix $\beta_k$ for all $k=1 \dots K$, this function becomes smooth from $p_k$. For example, we can take $\beta_k$ from point $\bm P = \bm p_1$ as $p_l=\frac{P/L}{\|\bm {w'}_l\|^2}$. Next, write $\textrm{SINR}$ similar to   Eq.~\eqref{SINR layer} without interference as $\textrm{SINR}_{kl}=\frac{{p}_l}{\|\bm g_l\|^2\sigma^2}$ ($\bm g_l$ does not depend on ${p}_l$). 
We can write Lagrangian for the problem~\ref{First formulation} (a):
\begin{equation}
    \mathcal{L}=-\sum_{k=1}^{K} L_k\ln\left (1 - \beta_l \log \left(
\frac {1} {L_k} \sum \limits_ {j = 1} ^ {L_k} \exp \left(- \dfrac {{p}_j} {\sigma^2\beta_l\|\bm g_j\|^2}\right)
\right)\right)+\lambda_i \left(\sum\limits_{l=1}^{L}(\|{{\bm w'}}_{l}\|^2{p}_{l}) - P\right)
\end{equation}

And its partial derivatives concerning ${p}_l$.
\begin{equation}\label{eq:func_derivative}
    \mathcal{L}^{'}_{{p}_l}=-\frac{\frac{1}{\sigma^2\|\bm g_l\|^2}x_l}
    {\left(1-\beta_l\ln\left(X_k\right)\right)X_k}+\sum\limits_{t=1}^{T}(\lambda_t|{w'}_{tl}|^2),
\end{equation}

where $x_l=\exp\left(-\frac{{p}_l}{\beta_l\sigma^2\|\bm g_l\|^2}\right)$ and $X_k=\frac{1}{L_k}\sum\limits_{i=1}^{L_k}\exp\left(-\frac{{p}_i}{\beta_l\sigma^2\|\bm g_i\|}\right)$.

We can write Karush–Kuhn–Tucker conditions:
	\begin{equation}\label{system_KKT_hard}
		\begin{cases} 
			\mathcal{L}^{'}_{{p}_l}=0, \quad l=1 \dots L
			\\
			\lambda_i \left(\sum\limits_{l=1}^{L}(\|{{\bm w'}}_{l}\|^2{p}_l)-P\right)=0, 
			\\
			\lambda_i\geqslant 0
		\end{cases} 
	\end{equation}

And its solution is (see proof in Appendix~\ref{sec:appendix_derive}):
\begin{equation}
    {p}_l=-\ln(x_l)\beta_k\sigma^2\|\bm g_l\|^{2}
    =
    \sigma^2\|\bm g_l\|^{2}
    \left(
        \frac{\frac{P}{\sigma^2L}+\frac{1}{L}\sum\limits_{v=1}^{L}\|\bm g_v\|^{2}\|{{\bm w'}}_v\|^2f_v
        }{ \frac{1}{L_k}\sum\limits_{v \in \mathcal{L}_k}
        \left(
            \|\bm g_v\|^{2}\|{{\bm w'}}_v\|^2
        \right)
        }-f_l
    \right),
    \label{p_1lambd_QAM_simple}
\end{equation}

where:
\begin{equation}
    f_l=\beta_k\ln\left(\frac{\|\bm g_l\|^{2} \|{{\bm w'}}_{l}\|^2}{\frac{1}{L_k}\sum\limits_{v \in \mathcal{L}_k}\|\bm g_v\|^{2} \|{{\bm w'}}_v\|^2} \right)+1
\end{equation}

\subsection{Simplified PA Problem with Per-Antenna Power Constraints}

\begin{theorem}
If $\bm W$ satisfies to the property (\ref{IRC 1}) and $G = G_{C}$, assuming model (\ref{SINR_eff}) of effective \textrm{SINR}, we can find a strict asymptotic solution of the second optimization problem 
\begin{equation}
U = \sum_n SE_n \rightarrow \max , \quad ||\bm w_t||^2 \leqslant \frac{P}{T},\;  t=1 \dots T.
\end{equation}
by solving the system of equations.
\end{theorem}
\begin{proof}
The problem~\eqref{First formulation} (b) can be reduced to a task
	\begin{equation} 
		\sum_{l=1}^L \log({p}_{l})  \rightarrow \max\limits_{\bm P}, \quad\text{subject to  } \sum\limits_{l=1}^{L}(|{{ w'}}_{tl}|^2{p}_l)\leqslant \frac{P}{T} \; \forall t = 1 \dots T
		\label{opt4}
	\end{equation}
	
    To solve it, we can use the Karush-–Kuhn-–Tucker conditions. Lagrangian has the form
	\begin{equation}
	\label{lagrang}
		\mathcal{L}=-\sum_{l=1}^L \log( {p}_{l}) + \sum\limits_{t=1}^{T}\left(\lambda_t \left(\sum\limits_{l=1}^{L}(|{{w'}}_{tl}|^2{p}_{l}) - \frac{P}{T} \right)\right).
	\end{equation}
	
	If ${p}_{l}$ and $\lambda_t$ are the optimum of the optimization problem, then they satisfy the following conditions
	\begin{equation}\label{system_KKT_2}
    	\begin{cases} 
    		{p}_l\sum\limits_{t=1}^{T}|{w'}_{tl}|^2\lambda_t=1, \quad l=1 \dots L\\
    		\lambda_t \left(\sum\limits_{l=1}^{L}(|{w'}_{tl}|^2{p}_l)-\frac{P}{T}\right)=0, \quad t=1 \dots T
    		\\
    		\lambda_t\geqslant 0, \quad t=1 \dots T
    	\end{cases} \Leftrightarrow
    	\begin{cases} 
    		\bm A^T \bm \lambda=1./{\bm p}\\
    		\bm \lambda .* \left(\bm A{\bm p}- \bm 1 \frac{P}{T}  \right)=0\\
    		\bm \lambda\geqslant 0
    	\end{cases} 
	\end{equation}

We take $ \bm A = \{a_ {ij} = | {w '} _ {ij} | ^ 2 \} $.

For geometric reasons, the original optimization problem has a solution; therefore, there is at least one solution to the system (\ref{system_KKT_2}).

The resulting system can be solved by brute force on the set of zeroed lambdas. Let's say we have non-zero $m$ lambdas. Consider the cases.

\begin{itemize}
    \item[1.] $m > L$, in this case the linear system $\sum\limits_{l=1}^{L}(|{w'}_{tl}|^2{p}_l)=\frac{P}{T}, t=1 \dots m$ will be inconsistent since the number of equations is greater than the number of unknowns ($m> L$) and the system itself (\ref{system_KKT_2}) will not have a solution.
    \item[2.] $m=L$, in this case the linear system $\sum\limits_{l=1}^{L}(|{w'}_{tl}|^2{p}_l)=\frac{P}{T}, t=1 \dots m$ has exactly one solution, and the system itself (\ref{system_KKT_2}) has at most one solution.
    \item[3.] $1<m<L$. This case reduces to the system of quadratic equations. If $\bm A'$ is a matrix consisting of rows of matrix $\bm A$ corresponding to nonzero lambdas, then 
    \begin{equation}
        \begin{cases} 
    		{\bm A'}^T\bm \lambda=1./{\bm p}\\
    		\left({\bm A'}{\bm p}- \bm 1 \frac{P}{T}\right)=0\\
    	\end{cases} \Rightarrow
    	\begin{cases} 
    	    (\bm A')^{\perp}(1./{\bm p})=0\\
    		{\bm A'}{\bm p}= \bm 1 \frac{P}{T}\\
    	\end{cases} 
	\end{equation}
	Here $\bm A'\in \mathbb{C}^{m\times L}$ and $(\bm A')^{\perp}\in \mathbb{C}^{(L-m)\times L}$ is the orthogonal complement to $\bm A'$.
    \item[4.] $m=1$, in this case, there are one nonzero lambda. Let $\lambda_i\neq0$ therefore 
    \begin{equation}
        \label{1lambd}
        {p}_l=\frac{1}{\lambda_i|{w'}_{il}|^2}=\frac{P}{TL|{w'}_{il}|^2}
    \end{equation}
\end{itemize}
	
\end{proof}

\subsection{EESM Model and Per-Antenna Power Constraints}

In this section, we combine two ideas of previous sections. We calculate Spectral efficiency~\eqref{SE_calc_beta_fixed} using exponential model~\eqref{SINR_eff} for the fixed $\beta$ value.
Using this, we can write Lagrangian for the problem~\ref{First formulation} (b) and its partial derivatives concerning $p_l$:
\begin{equation}
    \mathcal{L}=-\sum_{k=1}^{K} L_k\ln\left (1 - \beta_l \log \left(
\frac {1} {L_k} \sum \limits_ {j = 1} ^ {L_k} \exp \left(- \dfrac {{p}_j} {\sigma^2\beta_l\|\bm g_j\|^2}\right)
\right)\right)+\sum\limits_{t=1}^{T}\lambda_t \left(\sum\limits_{l=1}^{L}(|{w'}_{tl}|^2{p}_{l}) - \frac{P}{T}\right)
\end{equation}

\begin{equation}
    \mathcal{L}^{'}_{{p}_l}=-\frac{\frac{1}{\sigma^2\|\bm g_l\|^2}x_l}
    {\left(1-\beta_l\ln\left(X_k\right)\right)X_k}+\sum\limits_{t=1}^{T}(\lambda_t|{w'}_{tl}|^2)
\end{equation}

The Karush-–Kuhn-–Tucker conditions:
	\begin{equation}\label{system_KKT_hard_2}
		\begin{cases} 
			\mathcal{L}^{'}_{{p}_l}=0, \quad l=1 \dots L
			\\
			\lambda_t \left(\sum\limits_{l=1}^{L}(|{w'}_{tl}|^2{p}_l)-\frac{P}{T}\right)=0, \quad t=1 \dots T
			\\
			\lambda_t\geqslant 0, \quad t=1 \dots T
		\end{cases} 
	\end{equation}

The resulting system can be solved by brute force on the set of zeroed lambdas. Let's say we have non-zero $m$ lambdas. Consider the cases.
\begin{itemize}
    \item[1.] $m>L$, in this case the linear system $\sum\limits_{l=1}^{L}(|{w'}_{tl}|^2{p}_l)=\frac{P}{T}, t=1 \dots m$ will be inconsistent since the number of equations is greater than the number of unknowns ($m > L$) and the system itself (\ref{system_KKT_hard_2}) will not have a solution.
    \item[2.] $m=L$, in this case the linear system $\sum\limits_{l=1}^{L}(|{w'}_{tl}|^2{p}_l)=\frac{P}{T}, t=1 \dots m$ has exactly one solution, and the system itself (\ref{system_KKT_hard_2}) has at most one solution.
    \item[3.] $1<m<L$. If $T'$ is the set of indexes of nonzero lambda, then this case reduces to the system of following equations:
    \begin{equation}
    	\begin{cases} 
    		\mathcal{L}^{'}_{{p}_l}=0\\
    		\sum\limits_{l=1}^{L}(|{w'}_{tl}|^2{p}_l)=\frac{P}{T}, \quad t\in T'\\
    		\lambda_t\geqslant 0, \quad t\in T'
    	\end{cases} 
	\end{equation}
    \item[4.] $m=1$, in this case there are one nonzero lambda. Let $\lambda_i\neq0$ therefore 
    
\end{itemize}

\begin{equation}
    {p}_l=-\ln(x_l)\beta_k\sigma^2\|\bm g_l\|^{2}
    =
    \sigma^2\|\bm g_l\|^{2}
    \left(
        \frac{\frac{P}{\sigma^2TL}+\frac{1}{L}\sum\limits_{v=1}^{L}\|\bm g_v\|^{2}|{{w'}}_{iv}|^2f_v
        }{ \frac{1}{L_k}\sum\limits_{v \in \mathcal{L}_k}
        \left(
            \|\bm g_v\|^{2}|{{w'}}_{iv}|^2
        \right)
        }-f_l
    \right),
    \label{p_1lambd_QAM}
\end{equation}

where:
\begin{equation}
    f_l=\beta_k\ln\left(\frac{\|\bm g_l\|^{2} |{{w'}}_{il}|^2}{\frac{1}{L_k}\sum\limits_{v \in \mathcal{L}_k}\|\bm g_v\|^{2} |{{w'}}_{iv}|^2} \right)+1
\end{equation}



The proof is similar to the proof of Eq.~\eqref{p_1lambd_QAM_simple} which can be found in the Appendix.

\subsection{Heuristic Algorithms, based on KKT-analysis}\label{sec:algorithm}

\textcolor{black}{
In this section we proposed two algorithms for SE maximization in the case of PAPC of two different models of Effective SINR. The first Alg.~\ref{OPT Scheme} assume Geometrical Averaging model~\eqref{User SINR}, while the second Alg.~\ref{OPT Scheme QAM} uses the proper EESM model~\eqref{SINR_eff}}

\begin{algorithm}
\SetAlgoLined
\KwIn{Channel $\bm H = \bm U^ \mathrm H \bm S \bm V$ by Lemma~\ref{Main Decomposition}, precoding matrix~$\bm W(\bm V)$, station power~$P$, number of base station antennas $T$, noise $\sigma^2$;}

\textbf{Calculate} $ \bm A = \{a_ {ij} = | {w} _ {ij} | ^ 2 \} \in \mathbb R^{T \times L}$, where $\bm a_i \in \mathbb R^L$ is a row vector.\\

\textbf{Calculate} starting point $\bm p_1:(\bm p_1)_l=\frac{P}{TL\|\bm w_l\|^2}$ 
\\

\textbf{Calculate} the hyperplane on which the square of the starting point lies. The index of this hyperplane is the maximal row norm: $i(\bm p_1)=\arg \max_i\{\|(\bm W \diag( \bm p_1))_{i, :}\|\}$;

\textbf{Calculate} optimal point on this hyperplane 
$\bm p_2:({\bm p}_2)_l=\frac{P}{TL|w_{il}|^2}$

\eIf{$\bm p_2$ satisfies to  Per-Antenna Power Constraints}{
     \Return{$\bm W_{opt} = \bm W \diag(\bm p_2$)}}{

\textbf{Calculate} direction vector $\bm d=\bm p_2^2-\bm p_1^2$

\textbf{Calculate} first intersection $\bm p_{opt}^2$ on a beam \{$\bm p_1^2+\alpha \bm d|\alpha>0$\} with other hyperplanes : $\bm p_{opt}^2 = \bm p_1^2+\alpha_{opt}\bm d$ where $\alpha_{opt} = \min\{\alpha_i|\alpha_i=\frac{P/T-\bm a_i^T \bm p_1^2}{\bm a_i^T \bm d}>0\}$

\Return{$\bm W_{opt} = \bm W \diag(\bm p_{opt})$}
}
\caption{IM CD~--- Heuristic Intersection Method of Power Allocation using Conjugate Detection and effective \textrm{SINR} as the geometrical mean}
\label{OPT Scheme}
\end{algorithm}

\begin{algorithm}
\SetAlgoLined
\KwIn{Channel $\bm H = \bm U^ \mathrm H \bm S \bm V$ by Lemma~\ref{Main Decomposition}, station power $P$, noise $\sigma^2$;}
\textbf{Define} smooth precoding function $\bm W(\bm V)$; \\

\textbf{Define} smooth detection function $\bm G(\bm H, \bm W)$ using MMSE-IRC~\eqref{def:IRC} or CD~\eqref{Conjugate Detection};
\\
\textbf{Define} smooth target function $J(\bm P)$. For example, $J^{\textrm{SE}}(\bm P) = \textrm{SE}(\bm W' \bm P, \bm H, \bm G, \sigma^2)$ using
\eqref{Symbol SINR},~\eqref{Spectral Efficiency} and~\eqref{SINR_eff};
\\

\textbf{Calculate} $ \bm A = \{a_ {ij} = | {w} _ {ij} | ^ 2 \}$, where $\bm a_i \in \mathbb R^L$ is a row vector.\\

\textbf{Calculate} starting point $\bm p_1:(\bm p_1)_l=\frac{P}{TL\|\bm w_l\|^2}$ 
\\
\textbf{Calculate} the hyperplane on which the square of the starting point lies. The index of this hyperplane is the maximal row norm: $i(\bm p_1)=\arg \max_i\{\|(\bm W \diag( \bm p_1))_{i, :}\|\}$;

\textbf{Calculate} the optimal point on this hyperplane $\bm p_2=\arg\max(J^{\textrm{SE}}(\bm P))$~\eqref{p_1lambd_QAM}   

\If { $\min\limits_i({\bm p}_1)_i<0$}{
\Return{$\bm W_{opt} = \bm W \diag(\bm p_1)$}
}
\eIf{$\bm p_2$ satisfies to Per-Antenna Power Constraints}{
\Return{$\bm W_{opt} = \bm W \diag(\bm p_2)$}}{

\textbf{Calculate} direction vector $\bm d=\bm p_2^2-\bm p_1^2$

\textbf{Calculate} first intersection $\bm p_{opt}^2$ on a beam \{$\bm p_1^2+\alpha \bm d|\alpha>0$\} with other hyperplanes : $\bm p_{opt}^2 = \bm p_1^2+\alpha_{opt}\bm d$ where $\alpha_{opt} = \min\{\alpha_i|\alpha_i=\frac{P/T-\bm a_i^T \bm p_1^2}{\bm a_i^T \bm d}>0\}$ 

\Return{$\bm W_{opt} = \bm W \diag(\bm p_{opt})$}
}

\caption{IM CD and IM IRC~--- Heuristic Intersection Method of Power Allocation using MMSE IRC Detection and exponential effective \textrm{SINR}~\eqref{SINR_eff} with MCS-$\beta$ Tab.~\ref{tab:beta}}
\label{OPT Scheme QAM}
\end{algorithm}

\begin{figure}
    \centering
    \includegraphics[width=0.7\linewidth]{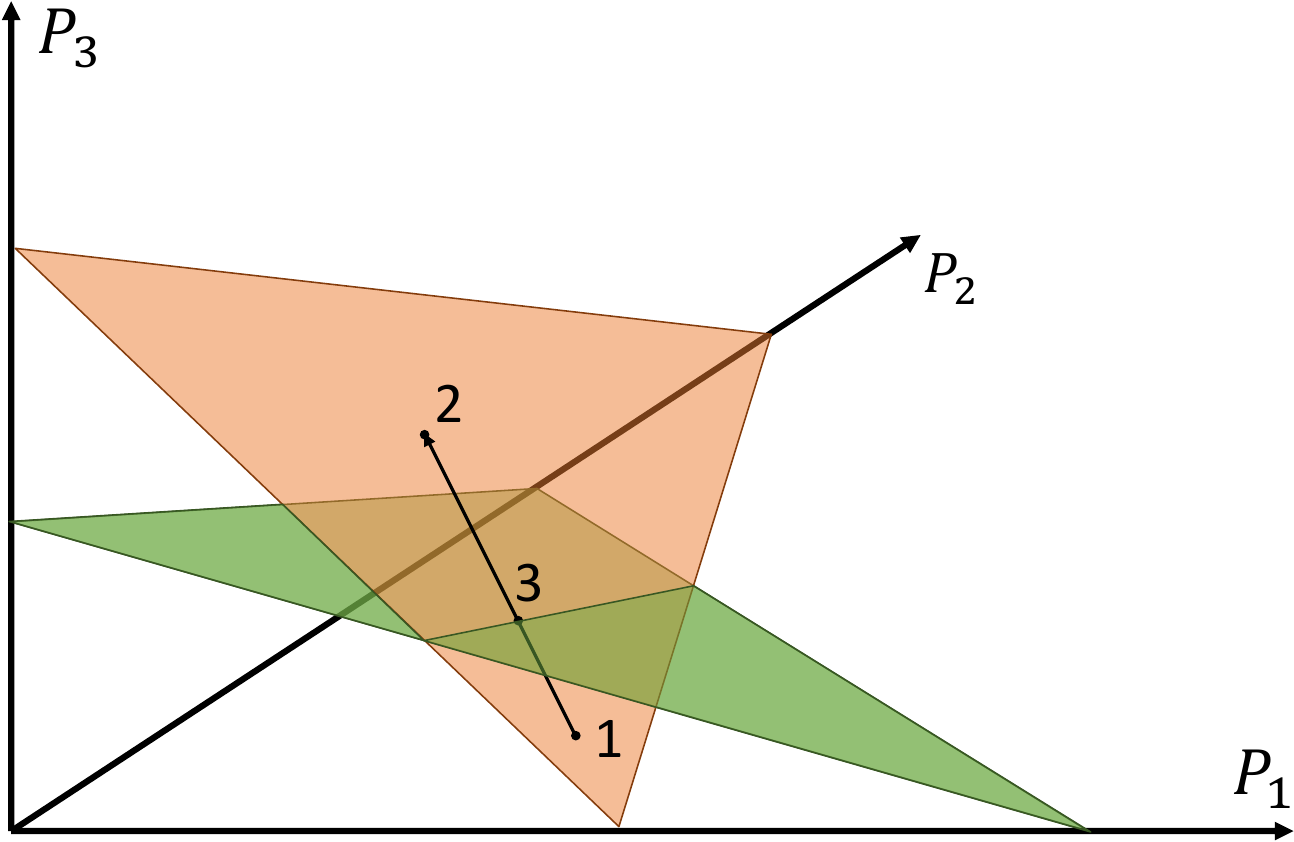}
    \caption{Geometrical illustration of the PA Intersection Method (Algs.~\ref{OPT Scheme} and Alg.~\ref{OPT Scheme QAM}) with a small layer space ($L=3$). Point~$1$ is the first approximation for solution of the algorithm and lies on the green hyperplane. Point~$2$ is the optimal Point on the green hyperplane, which we obtain as a solution to the Lagrange problem~\eqref{1lambd}. The red hyperplane contains the closest intersection of the beam from Point~$1$ to Point~$2$ concerning other hyperplanes. Point~$3$ is the intersection of the beam from Point~$1$ to Point~$2$ on the red hyperplane. If Point~$3$ lies between Points~$1$ and $2$, it is a solution of the algorithms, otherwise, a solution is Point~$2$.}
    \label{fig:geometrical_view}
\end{figure}

Both Alg.~\ref{OPT Scheme} and Alg.~\ref{OPT Scheme QAM} take equalizing powers as the first approximation of the vector $\bm p$ (see Point~$1$ on Fig.~\ref{fig:geometrical_view}). Then it finds the hyperplane on which the given point lies and searches on this hyperplane for the optimal (Point~$2$). To find the optimal point, we use the Eq.~\eqref{1lambd} for the Alg.~\ref{OPT Scheme} and by Eq.~\eqref{p_1lambd_QAM} for Alg.~\ref{OPT Scheme QAM}.

If the obtained point is satisfied with the Power Constrains, then this is the result of the algorithm. This point may not be satisfied with the Power Constraints. In this case, we construct a beam from the starting point to the optimal point. The first intersection with other hyperplanes (Point~$3$) is a result of the algorithms. The formula for Point~$2$ can be negative or zero. In this rare case, the result of the algorithm is Point~$1$. 




Fig.~\ref{fig:power_distr_im_method} shows the transmitted symbol powers using Alg.~\ref{OPT Scheme} (IM) compared to the EP method. The SINR values in dB of each layer are also given for comparison. It is shown that the SINR values increase for those symbols for which the power increases. And vice versa, the SINR decreases for those symbols for which the power decreases. The total precoding power increases with the use of Alg.~\ref{OPT Scheme} (IM).

\begin{figure}
    \centering
    \includegraphics[width=1\linewidth]{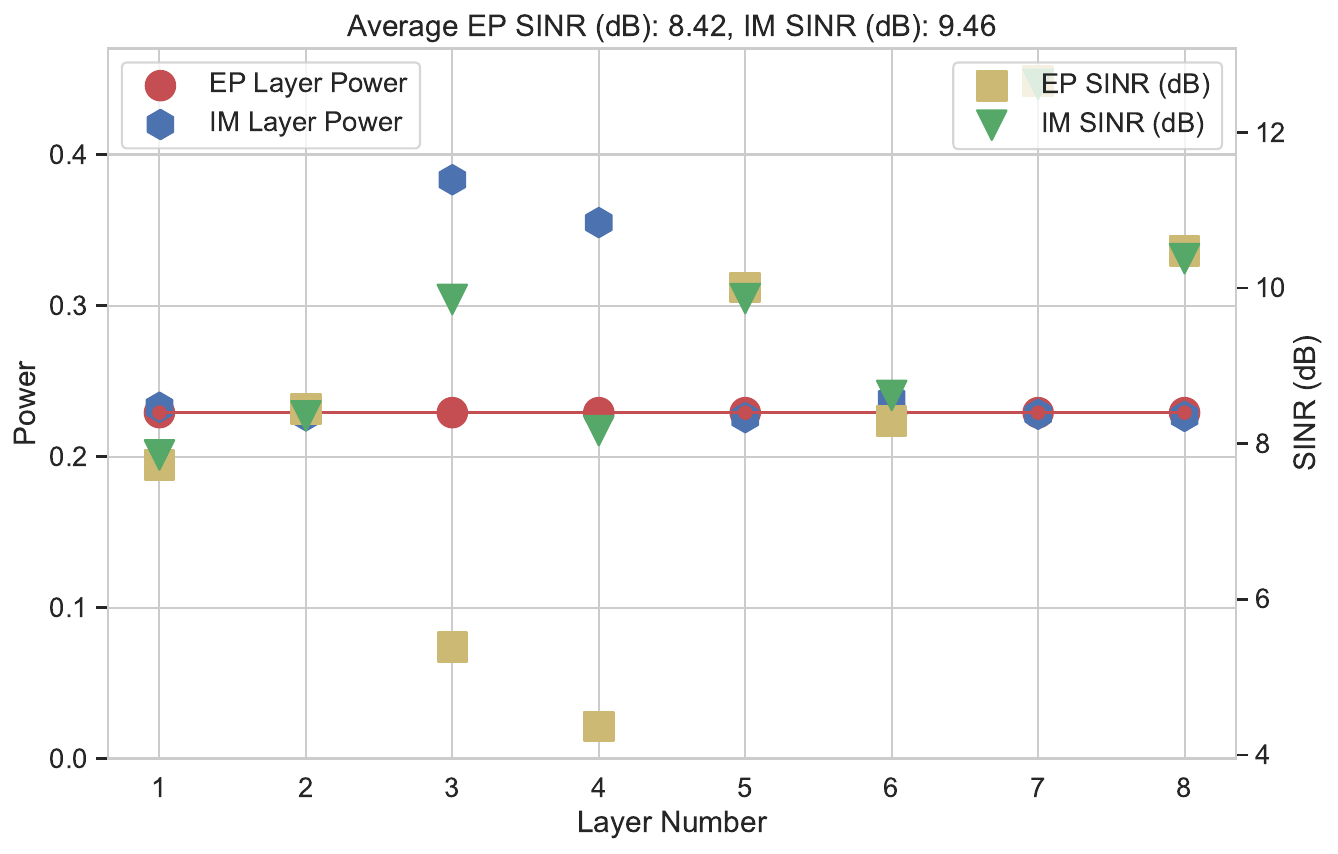}
    \caption{
        Power of the transmitted symbols (EP red circles and IM blue hexagons) and SINR of these symbols (EP yellow squares and IM green vertices) corresponding to Alg. EP and Alg.~\ref{OPT Scheme} IM. The EP method gives equal power to each transmitted symbol, which relates to equal yellow bars.
    }
    \label{fig:power_distr_im_method}
\end{figure}

\subsection{Computational Complexity}

In terms of complexity, the IM algorithms have the same order as the EP algorithm.

\begin{table}
\footnotesize
    \centering
    \caption{Complexity of the proposed Algs.~\ref{OPT Scheme} and~\ref{OPT Scheme QAM} with the number of summations, multiplications and use of special operations.}
    \begin{tabular}{|l|l|l|l|}
    \hline
    Variable & Summations & Multiplications & Special operations \\
    \hline
    $\bm A$ & $TL$ & $2TL$ & \\
    $\|\bm{w_l}\|^2, l=1,...,L $ & $(T-1)L$ & & \\
    $\bm{P_1}$ and $\bm{P_1}^2$& & $4L$ & \\
    $\bm W \bm{P_1}$ & & $2TL$ & \\
    $i(\bm{P_1})$ & $T(L-1)$ & $TL$ & $T-1$ comparisons\\ 
    $\bm{P_2}$ and $\bm{P_2}^2$ (Alg~\ref{OPT Scheme}) & & $5L$ & \\
    $\|\bm g^C_l\|^2, l=1,...,L$ & $2R_kL-L$ & $2R_kL$ &\\
    $\bm{P_2}$ and $\bm{P_2}^2$ (Alg~\ref{OPT Scheme QAM}) & $3L$ & $L(4L+8)$ & $L$ calculations of logarithm  \\
    $\bm W \bm{P_2}$ & & $2TL$ & \\
    $\bm d$ & $L$ & &\\
    $\alpha$ & $2TL$ & $T(L+1)$ & $T-1$ comparisons\\
    $\bm{P_{opt}}$ & $L$ & $L$ & \\
    $\bm{W_{opt}} = \bm W \bm{P_{opt}}$ & & $2TL$ & \\
    \hline
    Algorithm & Summations & Multiplications & Special operations \\
    \hline
    EP & $(2T-1)L$ & $ 4TL+4L$ & \\
    IM (Alg~\ref{OPT Scheme})& $ 5TL + L - T$ & $ 9TL+10L$ & $2T-2$ comparisons\\
    IM CD, IM IRC(Alg~\ref{OPT Scheme QAM})& $ 9TL + 2R_kL - T$ & $ (6T + 4L + 2R_k + 13)L$ & $2T-2$ comparisons, $L$ logs\\
    WF & $TL + 0.5L(L-1)$ & $4TL + L$ & $L$ comparisons, 1 sort \\
    \hline
    \end{tabular}\\

    \label{tab:alg_complexity}
\end{table}

In Tab.~\ref{tab:alg_complexity} computational complexity of each intermediate step of the algorithms \ref{OPT Scheme}, \ref{OPT Scheme QAM} and computational complexity of EP, IM, IM CD, IM IRC and WF algorithms are presented. For these algorithms we assume that we already calculate matrix $\bf W'$. For algorithms IM CD, IM IRC and WF we need precalculate matrices $\bm G^C$, $\bm G^{IRC}$ and $\bm S$ respectively. The difficulty of calculating some parts can be reduced. For example, when you calculating $P_{opt}$, you may not consider intersections with some hyperplanes. Note that for calculation of Alg. \ref{OPT Scheme QAM} we need to calculate matrix $\bf G$. 

The final complexity of the aforementioned algorithms is $\mathcal{O}(TL)$.  

\section{Simulation Results}\label{sec:results}

\subsection{Channel Dataset}

The datasets generated and analysed during the current study are available in the GitHub repository, \url{https://github.com/eugenbobrov/Power-Allocation-Algorithms-for-Massive-MIMO-Systems-with-Multi-Antenna-Users}

To generate channel coefficients, we use Quadriga~\cite{Quadriga}, open-source software for generating realistic radio channel impulse responses. We consider the Urban Non-Line-of-Sight~\cite{NLOS} scenarios. For each seed, we generate the random sets of user positions and compute channel matrices for the obtained configurations of users. Example of the random generation of users for Urban setup: there are two buildings, and the users are assigned to either a cluster in a building or to the ground near the building. The parameters of the experiments are listed in Table~\ref{table_example}. We describe the generation process in detail in our work~\cite{Conjugate}.



\subsection{Numerical Experiments}

We compare different PA algorithms based on RZF precoding. Primarily, the comparison involves precoding with the base power (BP) method~--- native method without PA, and the power equalization algorithm~\eqref{eq_pa}. Also, we consider some algorithms based on Karush-Kuhn-Tucker conditions~\eqref{system_KKT_2}. In Tab.~\ref{tab:pa_algorithms} algorithms with different parameters of the target optimization function, the power constraints and the starting point for intersection methods used for RZF method are presented. 

For reference we use the Power Allocation methods from the works of E. Bjornson et al., namely Equal Power (EP) and Water-Filling (WF) that are derived in assumption of Total Power Constraints (TPC). Proposed Intersection Methods (IM) are constructed to maximize Spectral Efficiency (SE) taking into account Per-Antenna Power Constraints (PAPC) and gives gains over the EP and WF methods in the specified region. Additionally, the IMs method can use WF solution as the starting point to achieve the cumulative gain in SE. This result is shown in Fig.~\ref{fig:geom_avg_se_gains}.

\begin{table}
\centering
    \caption{Review of the studied PA algorithms with their optimization function and assumed  constraints.}
    \begin{tabular}{|l|l|l|l|}
    \hline 
    Algorithm & Optimization Function & Constraints & Initialization\\
    \hline 
    EP & $\prod p_l \rightarrow \max$ & TPC & -\\
    IM & $\prod p_l \rightarrow \max$ & PAPC & EP\\
    WF & $SE(\bm G^{C}) \rightarrow \max$ & TPC & -\\
    IM CD & $SE(\bm G^C) \rightarrow \max$ & PAPC & EP\\
    IM IRC & $SE(\bm G^{IRC}) \rightarrow \max$ & PAPC & EP\\
    WF IM & $SE(\bm G^{C}) \rightarrow \max$ & PAPC & WF\\
    \hline 
    \end{tabular}
\label{tab:pa_algorithms}
\end{table}



In Figs.~\ref{fig:geom_avg_se_values} and~\ref{fig:avg_se_values_tab_1} we present an average SE~\eqref{Spectral Efficiency} from numerical simulations of the proposed Intersection Method~\eqref{OPT Scheme} IM and algorithm IM IRC with its modifications to MCS-$\beta$ model~\eqref{OPT Scheme QAM} IM CD and IM IRC and reference BP and EP methods. And in Figs.~\ref{fig:geom_avg_se_gains}-\ref{fig:avg_se_gains_tab_2} we present their gains over the reference EP method. Percentage gain means expressing the increase in SE value of the considered algorithm as a percentage compared to the baseline, in other words:
\begin{equation}
    SE \ Gain = \dfrac{SE_{considered} - SE_{baseline}}{SE_{baseline}}
\end{equation}

All Figs.~\ref{fig:geom_avg_se_values}-\ref{fig:avg_se_gains_tab_2} claim SE improvement of the proposed algorithms over the baseline EP method. Fig.~\ref{fig:geom_avg_se_gains} shows SE gain assuming Geometric Mean Effective SINR~\eqref{User SINR}, while Figs~\ref{fig:avg_se_gains_tab_1} and~\ref{fig:avg_se_gains_tab_2} assume the Exponential Averaging model~\eqref{SINR_eff}. Both the IM and IM IRC algorithms provide better power allocation (PA) under per-antenna power constraint (PAPC), which means better value of Spectral Efficiency~\eqref{Spectral Efficiency} of the obtained precoding in comparison to the BP and EP methods.

In Fig.~\ref{fig:power_cdf} we present the distribution of power allocated to different layers ($\|\bm w_l\|^2$) in case of PAPC when SU SINR is equal to 15dB. Cumulative distribution function (CDF) is calculated over transmitted layers.  
Here we see that IM majorizes both BP and EP methods in terms of power of layers (while still preserving PAPC), which is the main source of IM gains. In contrast, the WF method makes redistribution of power from UE with lower SINR to UE with higher SINR, which can be unfair and lead to blocking of cell-edge UE due to their poor contribution to the SE function. The WF IM (IM method applied to WF initial distribution) also majorizes WF and partially fixes its unfairness. 

Presented experiments claim that the proposed method IM outperforms the reference EP up to 5\% at the low \textrm{SUSINR} region ($<5$ dB) and up to 2\% at high ($>20$ dB). The modification of the algorithm IM IRC provides better results up to 6\% at the low \textrm{SUSINR} region. This is the result of better distribution of transmitted symbol powers (see example on Fig.~\ref{fig:power_distr_im_method}).

The proposed IM method in combination with widely-studied Water Filling (WF)~\cite{yu2004iterative} show a significant gain in spectral efficiency while using a similar computing time as the reference Equal Power (EP) solution (see Fig.~\ref{fig:geom_avg_se_gains}.)

The assumption that the noise-power ratio is close to zero was chosen that the Equal Power (EP) and Water-Filling (WF) method are close enough. Now we provide experiments both for EP and WF methods in Fig.~\ref{fig:geom_avg_se_gains} in a wide range of noise-power ratio. Although theoretical results stay correct only for close to zero noise-power, it helps to derive the Intersection Method (IM), which shows a good performance in a wide range of noise-power ratio.

Finally, it is experimentally proved out that the modification IM CD in case of both table 1 and table 2 MCS-$\beta$ values (see Tab.~\ref{tab:beta}) provides better results than IM. The difference in quality is clear in Gains of SE Figs.~\ref{fig:avg_se_gains_tab_1},~\ref{fig:avg_se_gains_tab_2}, which show that the performance improvement of Alg.~\ref{OPT Scheme QAM} is because of Alg.~\ref{OPT Scheme QAM} utilizes EESM Model~\ref{SINR_eff}.

\begin{figure}
    \centering
    \includegraphics[width=1\linewidth]{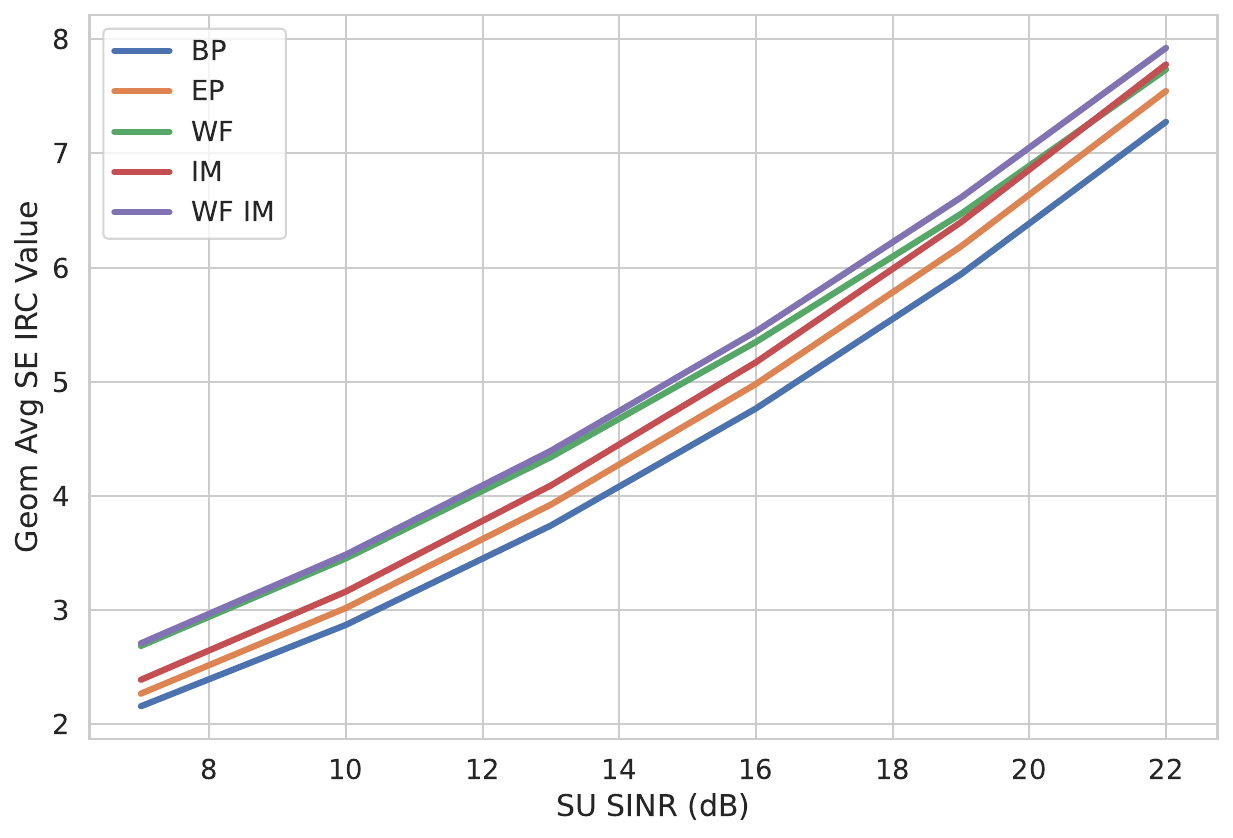}
    \caption{Average SE~\eqref{Spectral Efficiency} values Geometric Mean Effective SINR~\eqref{User SINR} and the different PA algorithms.}
    \label{fig:geom_avg_se_values}
\end{figure}

\begin{figure}
    \centering
    \includegraphics[width=\linewidth]{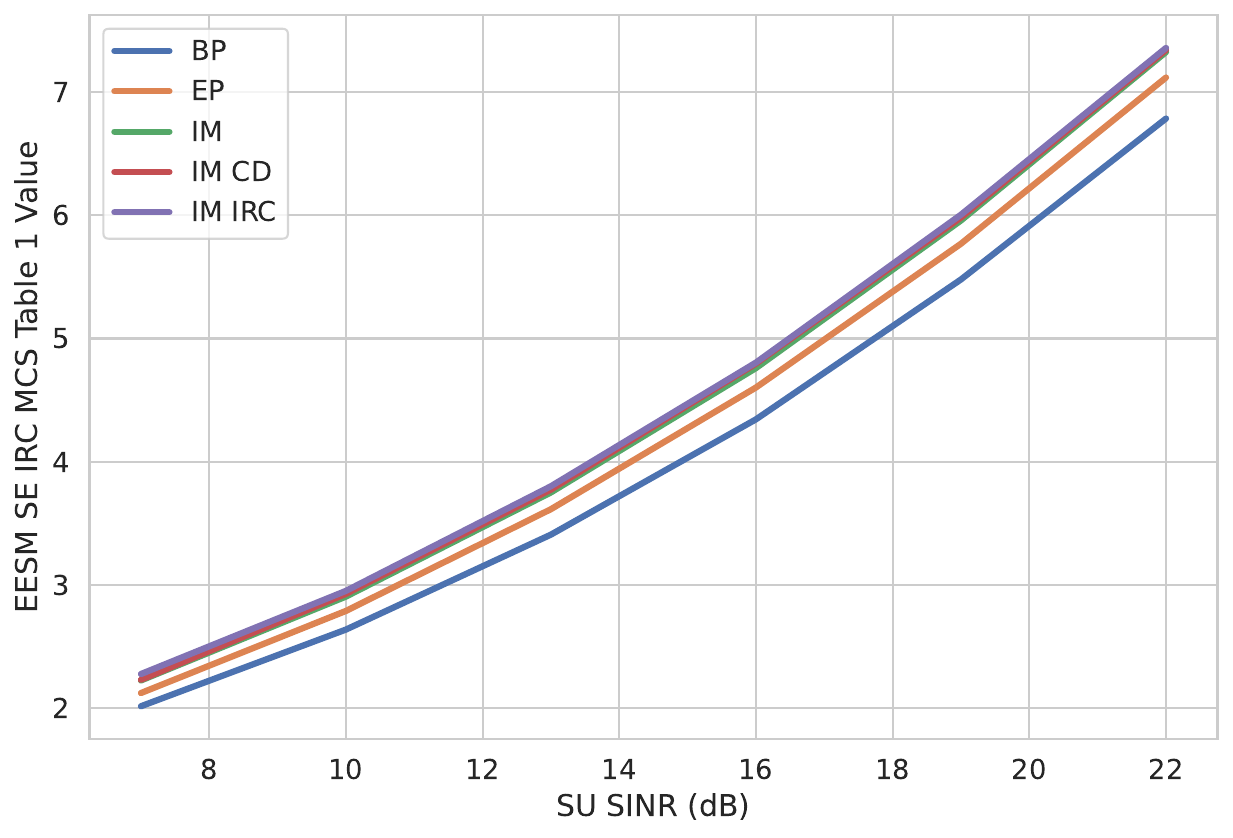}
    \caption{Average SE~\eqref{Spectral Efficiency} values using Exponential Averaging SINR~\eqref{SINR_eff} using table 1 MCS-$\beta$ values (see Tab.~\ref{tab:beta}) and the different PA algorithms. Using the table 2 MCS-$\beta$ values gives results similar to this plot.}
    \label{fig:avg_se_values_tab_1}
\end{figure}




\begin{figure}
    \centering
    \includegraphics[width=1\linewidth]{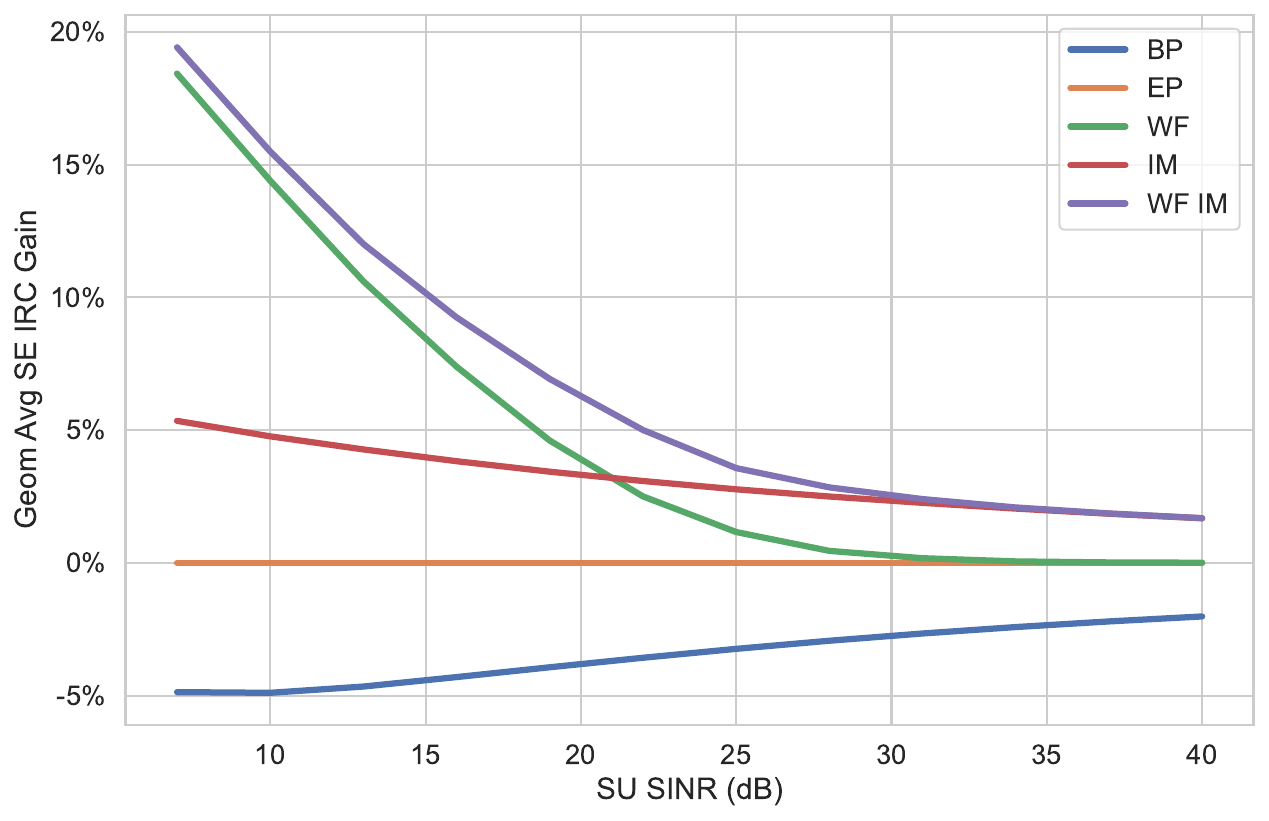}
    \caption{Average SE~\eqref{Spectral Efficiency} gains using Geometric Mean Effective SINR~\eqref{User SINR} and the different PA algorithms. } 
    \label{fig:geom_avg_se_gains}
\end{figure}


\begin{figure}
    \centering
    \includegraphics[width=\linewidth]{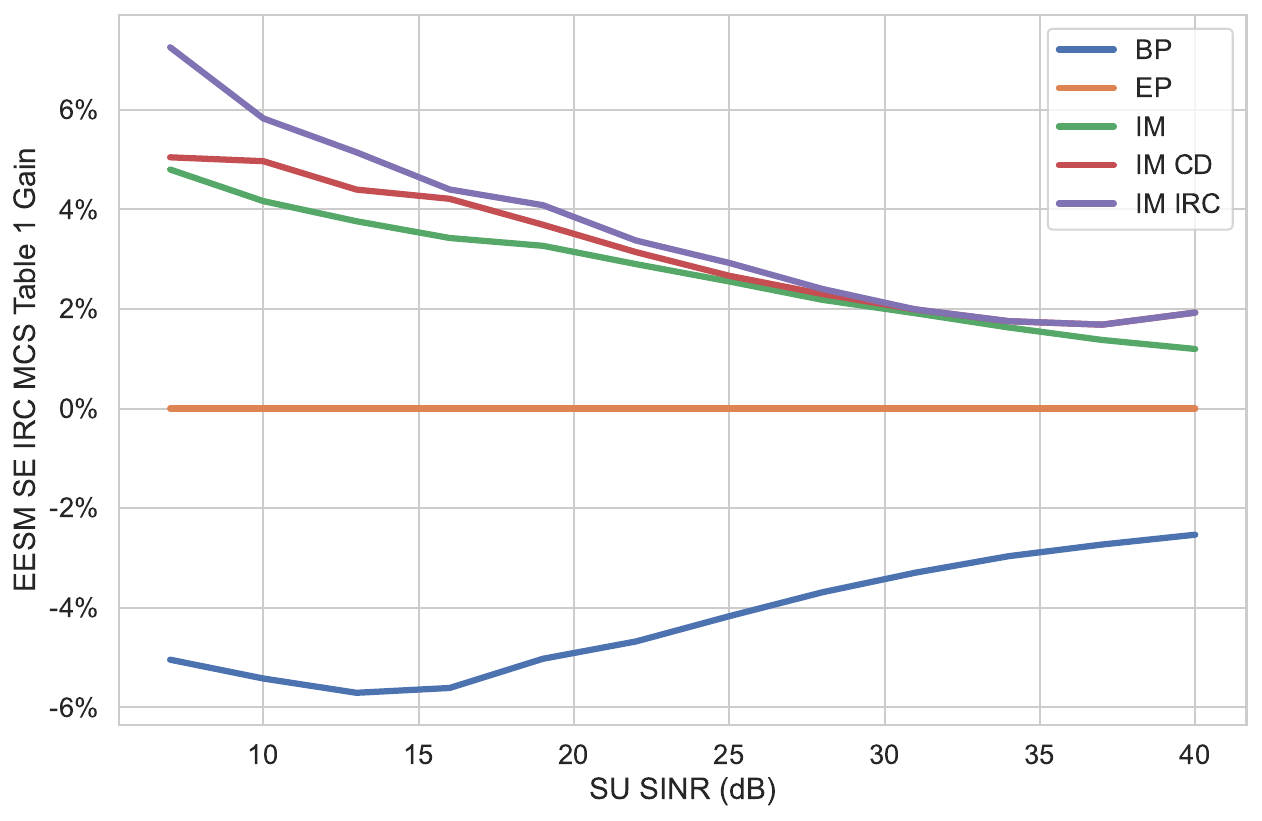}
    \caption{Average SE~\eqref{Spectral Efficiency} gains using Exponential Averaging SINR~\eqref{SINR_eff} using table 1 MCS-$\beta$ values (see Tab.~\ref{tab:beta}) and the different PA algorithms.}
    \label{fig:avg_se_gains_tab_1}
\end{figure}

\begin{figure}
    \centering
    \includegraphics[width=\linewidth]{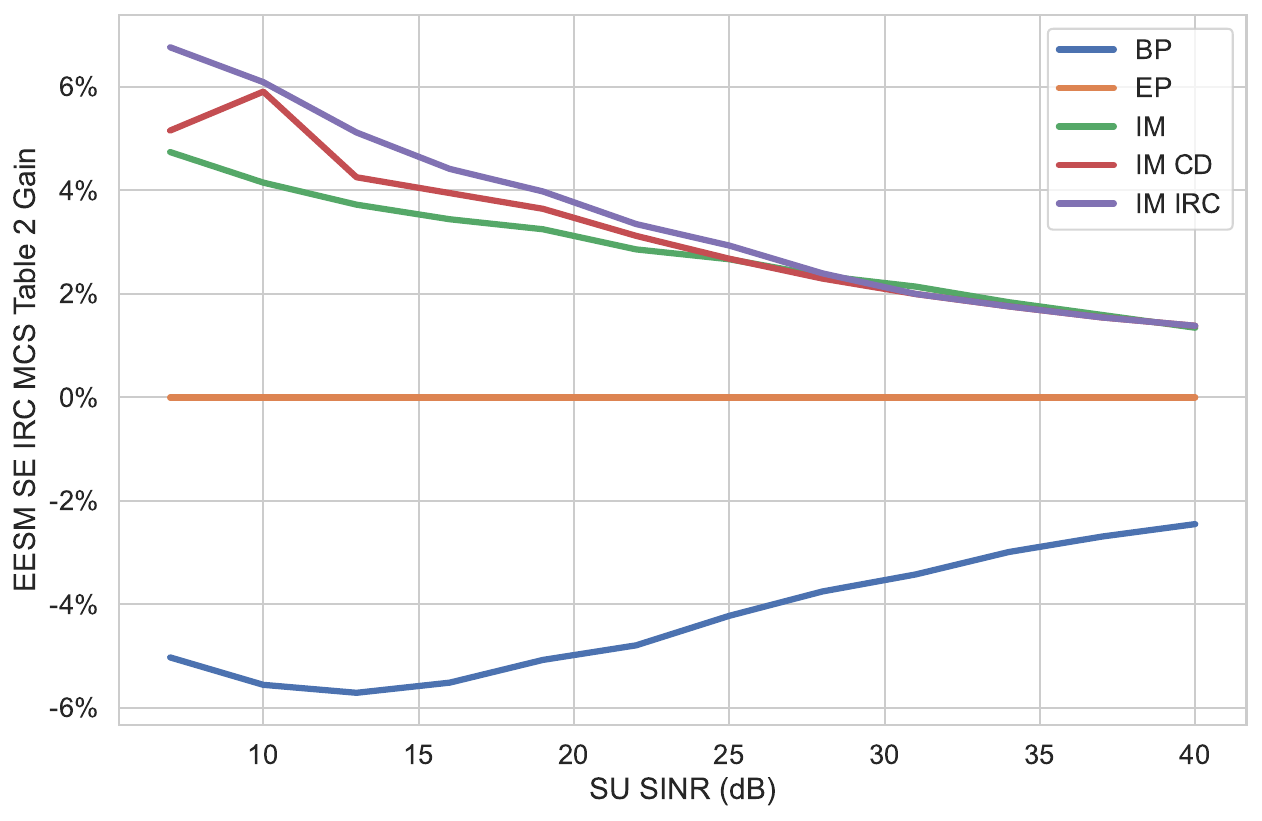}
    \caption{Average SE~\eqref{Spectral Efficiency} gains using Exponential Averaging SINR~\eqref{SINR_eff} using table 2 MCS-$\beta$ values (see Tab.~\ref{tab:beta}) and the different PA algorithms.}
    \label{fig:avg_se_gains_tab_2}
\end{figure}

\begin{figure}
    \centering
    \includegraphics[width=1\linewidth]{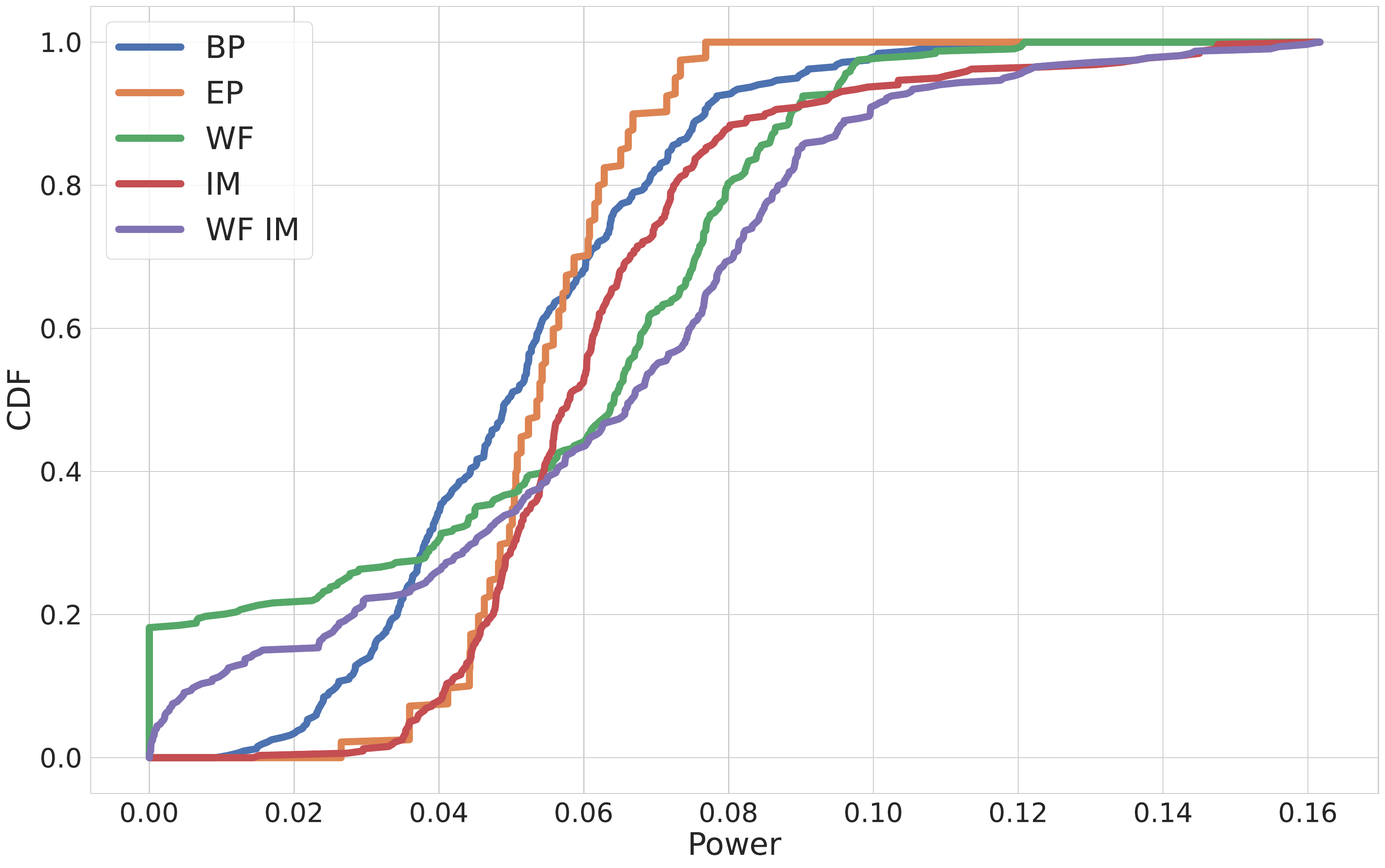}
    \caption{ The distribution of power of layers in the case of PAPC (on a set of scenarios with SU SINR = 15dB). }
    \label{fig:power_cdf}
\end{figure}




\section{Conclusions and Suggested Future Work}\label{sec:conclusion}


In this work, we study the power allocation (PA) problem of wireless MIMO systems with multi-antenna users. We simplify the initial problem using asymptotics of MMSE-IRC detection and SE function when noise and correlations are small. 
In the case of total power constraint (TPC) the simplified problem can be solved exactly and its solution is Equal Power (EP) distribution.
In the case of per-antenna power constraints (PAPC) simplified problem can be further equivalently reformulated as a Lagrange problem for which the Karush–Kuhn–Tucker conditions hold. 

Based on such analysis we propose low-complexity heuristic algorithms that provide sub-optimal solutions to the initial PA problem.
We study proposed Intersection Methods (IM)  on simulations using Quadriga and compare them with Equal Power and Water Filling reference algorithms. When simulated using Quadriga, the proposed IM methods combined with the widely studied Water Filling (WF) show a significant gain in SE using similar computational time compared to the EP baseline solution and allow improving the quality of MIMO systems in the future. Analyzing the CDF of power of layers we show that proposed IM methods majorize considered reference algorithms, provide more power under realistic Per-Antenna Power Constraints (PAPC) constraints and by this way improve Spectral Efficiency.

Since the main focus of this paper is the analytical study of the PA methods, we assume that the base station has perfect channel measurements and neglect all other potential hardware impairments. Nevertheless, the robustness of the noise to a given measurement keeps the current results asymptotically correct and can be carefully considered in future work. 
There are other possible direction of the future work. 
Firstly, future work can include a detailed study of PA algorithms, taking into account BLER performance with realistic 5G LDPC coding (e.g. using physical communication system level simulators such as Sionna~\cite{hoydis2022sionna}) rather then approximate effective SINR models such as EESM.
Secondly, the more complicated system model considering multiple RBs can be of interest. 
Thirdly, proposed IM algorithm can be perhaps further improved: improvement of SE can be realized with increasing the complexity of the algorithm, or otherwise, the complexity can be decreased with small decreasing of the SE.

\bibliographystyle{unsrt}
\bibliography{interacttfssample}

\begin{thebibliography}{10}

\bibitem{5G}
Jeffrey~G Andrews, Stefano Buzzi, Wan Choi, Stephen~V Hanly, Angel Lozano,
  Anthony~CK Soong, and Jianzhong~Charlie Zhang.
\newblock What will {5G} be?
\newblock {\em IEEE Journal on selected areas in communications},
  32(6):1065--1082, 2014.

\bibitem{marzetta2010noncooperative}
Thomas~L Marzetta.
\newblock Noncooperative cellular wireless with unlimited numbers of base
  station antennas.
\newblock {\em IEEE transactions on wireless communications}, 9(11):3590--3600,
  2010.

\bibitem{ge2016multi}
Xiaohu Ge, Ran Zi, Haichao Wang, Jing Zhang, and Minho Jo.
\newblock Multi-user massive {MIMO} communication systems based on irregular
  antenna arrays.
\newblock {\em IEEE Transactions on Wireless Communications}, 15(8):5287--5301,
  2016.

\bibitem{le2007multihop}
Long Le and Ekram Hossain.
\newblock Multihop cellular networks: Potential gains, research challenges, and
  a resource allocation framework.
\newblock {\em IEEE Communications Magazine}, 45(9):66--73, 2007.

\bibitem{phan2009power}
Khoa~T Phan, Tho Le-Ngoc, Sergiy~A Vorobyov, and Chintha Tellambura.
\newblock Power allocation in wireless multi-user relay networks.
\newblock {\em IEEE Transactions on Wireless Communications}, 8(5):2535--2545,
  2009.

\bibitem{EE}
Hien~Quoc Ngo, Erik~G Larsson, and Thomas~L Marzetta.
\newblock Energy and spectral efficiency of very large multiuser {MIMO}
  systems.
\newblock {\em IEEE Transactions on Communications}, 61(4):1436--1449, 2013.

\bibitem{ZF_MRT}
Tebe Parfait, Yujun Kuang, and Kponyo Jerry.
\newblock Performance analysis and comparison of {ZF} and {MRT} based downlink
  massive {MIMO} systems.
\newblock In {\em 2014 sixth international conference on ubiquitous and future
  networks (ICUFN)}, pages 383--388. IEEE, 2014.

\bibitem{RZF2}
Jiankang Zhang, Sheng Chen, Robert~G Maunder, Rong Zhang, and Lajos Hanzo.
\newblock Regularized zero-forcing precoding-aided adaptive coding and
  modulation for large-scale antenna array-based air-to-air communications.
\newblock {\em IEEE Journal on Selected Areas in Communications},
  36(9):2087--2103, 2018.

\bibitem{Survey2017}
Nusrat Fatema, Guang Hua, Yong Xiang, Dezhong Peng, and Iynkaran Natgunanathan.
\newblock Massive {MIMO} linear precoding: A survey.
\newblock {\em IEEE systems journal}, 12(4):3920--3931, 2017.

\bibitem{Survey2015}
Kan Zheng, Long Zhao, Jie Mei, Bin Shao, Wei Xiang, and Lajos Hanzo.
\newblock Survey of large-scale {MIMO} systems.
\newblock {\em IEEE Communications Surveys \& Tutorials}, 17(3):1738--1760,
  2015.

\bibitem{EZF19}
Sunil Dhakal.
\newblock High rate signal processing schemes for correlated channels in {5G}
  networks.
\newblock 2019.

\bibitem{Bjornson_tb_17}
Emil Bj{\"o}rnson, Jakob Hoydis, and Luca Sanguinetti.
\newblock Massive {MIMO} networks: Spectral, energy, and hardware efficiency.
\newblock {\em Foundations and Trends in Signal Processing}, 11(3-4):154--655,
  2017.

\bibitem{Tse_tb_05}
David Tse and Pramod Viswanath.
\newblock {\em Fundamentals of wireless communication}.
\newblock Cambridge university press, 2005.

\bibitem{boccardi2006optimum}
Federico Boccardi and Howard Huang.
\newblock Optimum power allocation for the {MIMO-BC} zero-forcing precoder with
  per-antenna power constraints.
\newblock In {\em 2006 40th Annual Conference on Information Sciences and
  Systems}, pages 504--504. IEEE, 2006.

\bibitem{yu2006uplink}
Wei Yu.
\newblock Uplink-downlink duality via minimax duality.
\newblock {\em IEEE Transactions on Information Theory}, 52(2):361--374, 2006.

\bibitem{bjornson2013optimal}
Emil Bj{\"o}rnson and Eduard Jorswieck.
\newblock {\em Optimal resource allocation in coordinated multi-cell systems}.
\newblock Now Publishers Inc, 2013.

\bibitem{deng2005power}
Xitirnin Deng and Alexander~M Haimovich.
\newblock Power allocation for cooperative relaying in wireless networks.
\newblock {\em IEEE Communications Letters}, 9(11):994--996, 2005.

\bibitem{host2005capacity}
Anders Host-Madsen and Junshan Zhang.
\newblock Capacity bounds and power allocation for wireless relay channels.
\newblock {\em IEEE transactions on Information Theory}, 51(6):2020--2040,
  2005.

\bibitem{liang2005gaussian}
Yingbin Liang and Venugopal~V Veeravalli.
\newblock Gaussian orthogonal relay channels: Optimal resource allocation and
  capacity.
\newblock {\em IEEE Transactions on Information Theory}, 51(9):3284--3289,
  2005.

\bibitem{zhao2006improving}
Yi~Zhao, Raviraj Adve, and Teng~Joon Lim.
\newblock Improving amplify-and-forward relay networks: optimal power
  allocation versus selection.
\newblock In {\em 2006 ieee international symposium on information theory},
  pages 1234--1238. IEEE, 2006.

\bibitem{nguyen2011power}
Duy~HN Nguyen and Ha~H Nguyen.
\newblock Power allocation in wireless multiuser multi-relay networks with
  distributed beamforming.
\newblock {\em IET communications}, 5(14):2040--2051, 2011.

\bibitem{sanguinetti2018deep}
Luca Sanguinetti, Alessio Zappone, and Merouane Debbah.
\newblock Deep learning power allocation in massive {MIMO}.
\newblock In {\em 2018 52nd Asilomar conference on signals, systems, and
  computers}, pages 1257--1261. IEEE, 2018.

\bibitem{van2020joint}
Trinh Van~Chien, Emil Bj{\"o}rnson, and Erik~G Larsson.
\newblock Joint power allocation and load balancing optimization for
  energy-efficient cell-free massive {MIMO} networks.
\newblock {\em IEEE Transactions on Wireless Communications},
  19(10):6798--6812, 2020.

\bibitem{SpatialCorrelation}
Emil Bj{\"o}rnson, Eduard Jorswieck, and Bjorn Ottersten.
\newblock Impact of spatial correlation and precoding design in {OSTBC} {MIMO}
  systems.
\newblock {\em IEEE Transactions on Wireless Communications}, 9(11):3578--3589,
  2010.

\bibitem{SVD}
Liang Sun and Matthew~R McKay.
\newblock Eigen-based transceivers for the {MIMO} broadcast channel with
  semi-orthogonal user selection.
\newblock {\em IEEE Transactions on Signal Processing}, 58(10):5246--5261,
  2010.

\bibitem{SINR_eff_model}
Zakaria Hanzaz and Hans~Dieter Schotten.
\newblock Analysis of effective {SINR} mapping models for {MIMO} {OFDM} in
  {LTE} system.
\newblock pages 1509--1515, 2013.

\bibitem{mohajer2022heterogeneous}
Amin Mohajer, Mahya~Sam Daliri, A~Mirzaei, A~Ziaeddini, M~Nabipour, and Maryam
  Bavaghar.
\newblock Heterogeneous computational resource allocation for noma: Toward
  green mobile edge-computing systems.
\newblock {\em IEEE Transactions on Services Computing}, 2022.

\bibitem{nikjoo2018novel}
Faramarz Nikjoo, Abbas Mirzaei, and Amin Mohajer.
\newblock A novel approach to efficient resource allocation in noma
  heterogeneous networks: Multi-criteria green resource management.
\newblock {\em Applied Artificial Intelligence}, 32(7-8):583--612, 2018.

\bibitem{mohajer2022energy}
Amin Mohajer, F~Sorouri, A~Mirzaei, A~Ziaeddini, K~Jalali Rad, and Maryam
  Bavaghar.
\newblock Energy-aware hierarchical resource management and backhaul traffic
  optimization in heterogeneous cellular networks.
\newblock {\em IEEE Systems Journal}, 16(4):5188--5199, 2022.

\bibitem{Quadriga}
Stephan Jaeckel, Leszek Raschkowski, Kai B{\"o}rner, and Lars Thiele.
\newblock {QuaDRiGa}: A {3-D} multi-cell channel model with time evolution for
  enabling virtual field trials.
\newblock {\em IEEE Transactions on Antennas and Propagation},
  62(6):3242--3256, 2014.

\bibitem{Aitken}
Alexander~C Aitken.
\newblock {IV}.—{On} least squares and linear combination of observations.
\newblock {\em Proceedings of the Royal Society of Edinburgh}, 55:42--48, 1936.

\bibitem{Zaidi_tb_18}
Ali Zaidi, Fredrik Athley, Jonas Medbo, Ulf Gustavsson, Giuseppe Durisi, and
  Xiaoming Chen.
\newblock {\em {5G} Physical Layer: principles, models and technology
  components}.
\newblock Academic Press, 2018.

\bibitem{Conjugate}
Evgeny Bobrov, Boris Chinyaev, Viktor Kuznetsov, Hao Lu, Dmitrii Minenkov,
  Sergey Troshin, Daniil Yudakov, and Danila Zaev.
\newblock Adaptive regularized zero-forcing beamforming in {Massive} {MIMO}
  with multi-antenna users.
\newblock {\em arXiv preprint arXiv:2107.00853}, 2021.

\bibitem{RankSelection}
Nurul~H Mahmood, Gilberto Berardinelli, Fernando~ML Tavares, Mads Lauridsen,
  Preben Mogensen, and Kari Pajukoski.
\newblock An efficient rank adaptation algorithm for cellular {MIMO} systems
  with {IRC} receivers.
\newblock In {\em 2014 IEEE 79th Vehicular Technology Conference (VTC Spring)},
  pages 1--5. IEEE, 2014.

\bibitem{VAE}
Evgeny Bobrov, Alexander Markov, and Dmitry Vetrov.
\newblock Variational autoencoders for studying the manifold of precoding
  matrices with high spectral efficiency.
\newblock {\em arXiv preprint arXiv:2111.15626}, 2021.

\bibitem{Joham_RZF}
Michael Joham, Wolfgang Utschick, and Josef~A Nossek.
\newblock Linear transmit processing in {MIMO} communications systems.
\newblock {\em IEEE Transactions on signal Processing}, 53(8):2700--2712, 2005.

\bibitem{nguyen2014mmse}
Duy~HN Nguyen and Tho Le-Ngoc.
\newblock Mmse precoding for multiuser miso downlink transmission with
  non-homogeneous user snr conditions.
\newblock {\em EURASIP Journal on Advances in Signal Processing}, 2014:1--12,
  2014.

\bibitem{PrecodingDetection}
Shuying Shi, Martin Schubert, and Holger Boche.
\newblock Downlink {MMSE} transceiver optimization for multiuser {MIMO}
  systems: Duality and sum-{MSE} minimization.
\newblock {\em IEEE Transactions on Signal Processing}, 55(11):5436--5446,
  2007.

\bibitem{MMSE}
Ahmed~Hesham Mehana and Aria Nosratinia.
\newblock Diversity of {MMSE} {MIMO} receivers.
\newblock {\em IEEE Transactions on information theory}, 58(11):6788--6805,
  2012.

\bibitem{MMSE2}
Dirk Wubben, Ronald Bohnke, Volker Kuhn, and K-D Kammeyer.
\newblock Near-maximum-likelihood detection of {MIMO} systems using
  {MMSE}-based lattice-reduction.
\newblock In {\em 2004 IEEE International Conference on Communications (IEEE
  Cat. No. 04CH37577)}, volume~2, pages 798--802. IEEE, 2004.

\bibitem{IRC}
Bin Ren, Yingmin Wang, Shaohui Sun, Yawen Zhang, Xiaoming Dai, and Kai Niu.
\newblock Low-complexity {MMSE-IRC} algorithm for uplink massive {MIMO}
  systems.
\newblock {\em Electronics Letters}, 53(14):972--974, 2017.

\bibitem{SINR}
Bin Wang, Yongyu Chang, and Dacheng Yang.
\newblock On the {SINR} in massive {MIMO} networks with {MMSE} receivers.
\newblock {\em IEEE Communications Letters}, 18(11):1979--1982, 2014.

\bibitem{SE}
Sergio Verd{\'u}.
\newblock Spectral efficiency in the wideband regime.
\newblock {\em IEEE Transactions on Information Theory}, 48(6):1319--1343,
  2002.

\bibitem{Bjornson}
Emil Bj{\"o}rnson, Mats Bengtsson, and Bj{\"o}rn Ottersten.
\newblock Optimal multiuser transmit beamforming: A difficult problem with a
  simple solution structure [lecture notes].
\newblock {\em IEEE Signal Processing Magazine}, 31(4):142--148, 2014.

\bibitem{QAM64}
Sandra Lagen, Kevin Wanuga, Hussain Elkotby, Sanjay Goyal, Natale Patriciello,
  and Lorenza Giupponi.
\newblock New radio physical layer abstraction for system-level simulations of
  {5G} networks.
\newblock In {\em ICC 2020-2020 IEEE International Conference on Communications
  (ICC)}, pages 1--7. IEEE, 2020.

\bibitem{brueninghaus2005link}
Karsten Brueninghaus, David Astely, Thomas Salzer, Samuli Visuri, Angeliki
  Alexiou, Stephan Karger, and G-A Seraji.
\newblock Link performance models for system level simulations of broadband
  radio access systems.
\newblock In {\em 2005 IEEE 16th international symposium on personal, indoor
  and mobile radio communications}, volume~4, pages 2306--2311. IEEE, 2005.

\bibitem{6678684}
Jobin Francis and Neelesh~B. Mehta.
\newblock Eesm-based link adaptation in point-to-point and multi-cell ofdm
  systems: Modeling and analysis.
\newblock {\em IEEE Transactions on Wireless Communications}, 13(1):407--417,
  2014.

\bibitem{QNS}
Evgeny Bobrov, Dmitry Kropotov, Sergey Troshin, and Danila Zaev.
\newblock {L-BFGS} precoding optimization algorithm for massive {MIMO} systems
  with multi-antenna users, 2021.

\bibitem{NLOS}
Frode Bohagen, P{\aa}l Orten, and GE~Oien.
\newblock Construction and capacity analysis of high-rank line-of-sight {MIMO}
  channels.
\newblock In {\em IEEE Wireless Communications and Networking Conference,
  2005}, volume~1, pages 432--437. IEEE, 2005.

\bibitem{yu2004iterative}
Wei Yu, Wonjong Rhee, Stephen Boyd, and John~M Cioffi.
\newblock Iterative water-filling for gaussian vector multiple-access channels.
\newblock {\em IEEE Transactions on Information Theory}, 50(1):145--152, 2004.

\bibitem{hoydis2022sionna}
Jakob Hoydis, Sebastian Cammerer, Fay{\c{c}}al~Ait Aoudia, Avinash Vem,
  Nikolaus Binder, Guillermo Marcus, and Alexander Keller.
\newblock Sionna: An open-source library for next-generation physical layer
  research.
\newblock {\em arXiv preprint arXiv:2203.11854}, 2022.

\bibitem{MCS}
Evgeny Bobrov, Dmitry Kropotov, and Hao Lu.
\newblock Massive {MIMO} adaptive modulation and coding using online deep
  learning algorithm.
\newblock {\em IEEE Communications Letters}, 2021.

\bibitem{MCSTables}
{TSG RAN; NR;}.
\newblock Physical layer procedures for data (release 16) v16.0.0.
\newblock {\em 3GPP TS 38.214}, 2019.

\end{thebibliography}

\section*{Abbreviations}
\begin{tabular}{c|c}
    ARZF & Adaptive Regularized Zero-Forcing \\
    BP & Baseline Power \\
    CD & Conjugate Detection \\
    CDF & Cumulative Density Function \\
    CSI & Channel state information \\
    EESM & Exponential Effective SINR Mapping \\
    EP & Equal Power \\
    ESM & Effective SINR Mapping \\
    IM & Intersection Method \\
    IRC & Interference Rejection Combiner \\
    LOS & Line-of-Sight \\
    MCS & Modulation and Coding Scheme \\
    MIMO & Multiple-input multiple-output  \\
    MMSE & Minimum Mean Squared Error \\
    MRT & Maximum Ratio Transmission \\
    MSE & Mean Squared Error \\
    NLOS & Non-Line-of-Sight \\
    OFDM & Orthogonal Frequency-Division Multiplexing \\
    PA & Power Allocation \\
    PAPC & Per-Antenna Power Constraints \\
    PHY & Physical Layer \\
    RZF & Regularized Zero-Forcing \\
    SE & Spectral Efficiency \\
    SINR & Signal-to-Interference-and-Noise \\
    SVD & Singular-Value-Decomposition \\
    TDD & Time division duplex \\
    TPC & Total Power Constraints \\
    UE & User equipment \\
    WF & Water Filling \\
    ZF & Zero-Forcing
\end{tabular}

\section*{Statements and Declarations}

\subsection*{Acknowledgements}
Authors are grateful to Irina Basieva, Lu Hao, Dmitri Shmelkin and Yue Zongdi for discussions and support.
Also authors appreciate valuable and constructive comments from unknown reviewers.

\subsection*{Funding}
The research was supported by Huawei Technologies.

\subsection*{Competing Interests}
The authors have no relevant financial or non-financial interests to disclose.

\subsection*{Data availability}
The datasets generated and analysed during the current study are available in the GitHub repository, \url{https://github.com/eugenbobrov/Power-Allocation-Algorithms-for-Massive-MIMO-Systems-with-Multi-Antenna-Users}

\section*{Appendix}

\subsection{Search of MCS-$\beta$ Effective SINR}

The values of $\beta$ for Modulation and Coding Scheme (MCS)~\cite{MCS} are taken from Tab.~\ref{tab:beta}. There are different $\beta$ values for different MCSes~\cite{QAM64}. The Table~\ref{tab:beta} shows $\beta$ values, which corresponds to Tables 5.1.3.1-1 to 5.1.3.1-2 in~\cite{MCSTables}. The MCS value depends on the radio quality and therefore on $\textrm{SINR}^{eff}_\beta$.





Thus, $\textrm{SINR}^{eff}_\beta$ can be found by simple iteration method on the equation~\eqref{SINR_eff}, initializing $\textrm{SINR}^{eff}_\beta$ by geometrical average using~\eqref{User SINR} and then taking $\beta = \beta (\textrm{MCS})$ from Tab.~\ref{tab:beta} and $ \textrm{MCS} = \textrm{MCS}(\textrm{SINR}^{eff}_\beta)$ from Tab.~\ref{tab:se}.

Also note that low values of $\textrm{SINR}^{eff}_\beta$ (up to~-5dB) indicate that the user is almost out of service, and high values of $\textrm{SINR}^{eff}_\beta$ (after~23dB) do not make much sense.

\begin{table}[!htb]
    \begin{minipage}{.5\linewidth}
        \centering
        \caption{Optimal $\beta$ values for each MCS.}
        \begin{tabular}{|l|l|l|}
        \hline
        MCS & $\beta$-table~1 & $\beta$-table~2 \\ \hline
        0   & 1.6    & 1.6     \\ \hline
        1   & 1.61   & 1.63    \\ \hline
        2   & 1.63   & 1.67    \\ \hline
        3   & 1.65   & 1.73    \\ \hline
        4   & 1.67   & 1.79    \\ \hline
        5   & 1.7    & 4.27    \\ \hline
        6   & 1.73   & 4.71    \\ \hline
        7   & 1.76   & 5.16    \\ \hline
        8   & 1.79   & 5.66    \\ \hline
        9   & 1.82   & 6.16    \\ \hline
        10  & 3.97   & 6.5     \\ \hline
        11  & 4.27   & 10.97   \\ \hline
        12  & 4.71   & 12.92   \\ \hline
        13  & 5.16   & 14.96   \\ \hline
        14  & 5.66   & 17.06   \\ \hline
        15  & 6.16   & 19.33   \\ \hline
        16  & 6.5    & 21.85   \\ \hline
        17  & 9.95   & 24.51   \\ \hline
        18  & 10.97  & 27.14   \\ \hline
        19  & 12.92  & 29.94   \\ \hline
        20  & 14.96  & 56.48   \\ \hline
        21  & 17.06  & 65      \\ \hline
        22  & 19.33  & 78.58   \\ \hline
        23  & 21.85  & 92.48   \\ \hline
        24  & 24.51  & 106.27  \\ \hline
        25  & 27.14  & 118.74  \\ \hline
        26  & 29.94  & 126.36  \\ \hline
        27  & 32.05  & 132.54  \\ \hline
        \end{tabular}
        \label{tab:beta}
    \end{minipage}%
    \begin{minipage}{.5\linewidth}
        \caption{Optimal SE values for each MCS.}

          \begin{tabular}{|l|c|c|}
    \hline
    MCS & \multicolumn{1}{l|}{SE-table 1} & \multicolumn{1}{l|}{SE-table 2} \\ \hline
    0   & 0.2344                          & 0.2344                          \\ \hline
    1   & 0.3066                          & 0.377                           \\ \hline
    2   & 0.377                           & 0.6016                          \\ \hline
    3   & 0.4902                          & 0.877                           \\ \hline
    4   & 0.6016                          & 1.1758                          \\ \hline
    5   & 0.7402                          & 1.4766                          \\ \hline
    6   & 0.877                           & 1.6953                          \\ \hline
    7   & 1.0273                          & 1.9141                          \\ \hline
    8   & 1.1758                          & 2.1602                          \\ \hline
    9   & 1.3262                          & 2.4063                          \\ \hline
    10  & 1.3281                          & 2.5703                          \\ \hline
    11  & 1.4766                          & 2.7305                          \\ \hline
    12  & 1.6953                          & 3.0293                          \\ \hline
    13  & 1.9141                          & 3.3223                          \\ \hline
    14  & 2.1602                          & 3.6094                          \\ \hline
    15  & 2.4063                          & 3.9023                          \\ \hline
    16  & 2.5703                          & 4.2129                          \\ \hline
    17  & 2.7305                          & 4.5234                          \\ \hline
    18  & 3.0293                          & 4.8164                          \\ \hline
    19  & 3.3223                          & 5.1152                          \\ \hline
    20  & 3.6094                          & 5.332                           \\ \hline
    21  & 3.9023                          & 5.5547                          \\ \hline
    22  & 4.2129                          & 5.8906                          \\ \hline
    23  & 4.5234                          & 6.2266                          \\ \hline
    24  & 4.8164                          & 6.5703                          \\ \hline
    25  & 5.1152                          & 6.9141                          \\ \hline
    26  & 5.332                           & 7.1602                          \\ \hline
    27  & 5.5547                          & 7.4063                          \\ \hline
    \end{tabular}
    \label{tab:se}
    \end{minipage} 
\end{table}

\subsection{Derivation of the eq.~\eqref{p_1lambd_QAM_simple}}\label{sec:appendix_derive}

From the identity~\eqref{eq:func_derivative} $\mathcal{L}^{'}_{{p}_l}=0$:
\begin{equation}\label{xl_calc}
    x_l=(1-\beta_k\ln(X_k))X_k \beta_k\sigma^2\|\bm g_l\|^2 \lambda_i\|{{\bm w'}}_{l}\|^2
\end{equation}

Taking average of~\eqref{xl_calc}:
\begin{equation} \label{Xk_calc}
    X_k=\frac{1}{L_k}\sum\limits_{l \in \mathcal{L}_k}x_l \Leftrightarrow X_k=(1-\beta_k\ln(X_k))X_k
    \frac{1}{L_k}\sum\limits_{l \in \mathcal{L}_k}\left(
    \sigma^2s_l^{-2}
    \lambda_i\|{{\bm w'}}_{l}\|^2
    \right)
\end{equation}

Dividing~\eqref{xl_calc} by~\eqref{Xk_calc} we get:
\begin{equation}
    \frac{x_l}{X_k}=\frac{\sigma^2s_l^{-2} \lambda_i\|{{\bm w'}}_{l}\|^2}{
    \frac{1}{L_k}\sum\limits_{v \in \mathcal{L}_k}\left(
    \sigma^2s_v^{-2}
    \lambda_i\|{{\bm w'}}_v\|^2
    \right)}=\frac{s_l^{-2} \|{{\bm w'}}_{l}\|^2}{
    \frac{1}{L_k}\sum\limits_{v \in \mathcal{L}_k}\left(
    s_v^{-2}
    \|{{\bm w'}}_v\|^2
    \right)}
    \label{x_l}
\end{equation}

From~\eqref{Xk_calc}:
\begin{equation}\label{fk_calc}
    X_k=\exp\left(\frac{1}{\beta_k}-\frac{1}{\beta_k
    \frac{1}{L_k}\sum\limits_{l \in \mathcal{L}_k}\left(
    \sigma^2s_l^{-2}
    \lambda_i\|{{\bm w'}}_{l}\|^2
    \right)}\right)
\end{equation}

From \eqref{x_l} and \eqref{fk_calc} we can derive:
\begin{equation}\label{x_l_final}
    x_l = \frac{s_l^{-2} \|{{\bm w'}}_{l}\|^2}{
    \frac{1}{L_k}\sum\limits_{v \in \mathcal{L}_k}\left(
    s_v^{-2}
    \|{{\bm w'}}_v\|^2
    \right)} \exp\left(\frac{1}{\beta_k}-\frac{1}{\beta_k
    \frac{1}{L_k}\sum\limits_{l \in \mathcal{L}_k}\left(
    \sigma^2s_l^{-2}
    \lambda_i\|{{\bm w'}}_{l}\|^2
    \right)}\right)
\end{equation}

Also we know that $x_l=\exp\left(-\frac{{p}_l}{\beta_l\sigma^2 s_l^{-2}}\right)$. So we know $p_l=-\beta_l\sigma^2 s_l^{-2}\ln\left(x_l\right)$ and we can substitute \eqref{x_l_final} in the $p_l$ expression.

Taking into account $\sum\limits_{l=1}^{L}(\|{{\bm w'}}_{l}\|^2{p}_l)=P$ we obtain:
\begin{multline*}
\label{1lambd_IRC}
    \sum\limits_{l=1}^{L}(\|{{\bm w'}}_{l}\|^2{p}_l)
    =-\sum\limits_{k=1}^{K}\sum\limits_{l \in \mathcal{L}_k}\left[\sigma^2s_l^{-2}\|{{\bm w'}}_{l}\|^2\left(1-\frac{1}{
    \frac{1}{L_{k}}\sum\limits_{v=1}^{L_{k}}\left(
    \sigma^2s_v^{-2}
    \lambda_i\|{{\bm w'}}_v\|^2
    \right)}\right)\right]-\\
    - \sum\limits_{l=1}^{L}\beta_k\sigma^2s_l^{-2}\|{{\bm w'}}_{l}\|^2\ln\left(\frac{s_l^{-2} \|{{\bm w'}}_{l}\|^2}{\frac{1}{L_{k}}\sum\limits_{v=1}^{L_{k}}s_v^{-2} \|{{\bm w'}}_v\|^2} \right)=\\
    =\lambda_i^{-1}L-\sum\limits_{l=1}^{L}\sigma^2s_l^{-2}\|{{\bm w'}}_{l}\|^2- \sum\limits_{l=1}^{L}\beta_k\sigma^2s_l^{-2}\|{{\bm w'}}_{l}\|^2\ln\left(\frac{s_l^{-2} \|{{\bm w'}}_{l}\|^2}{\frac{1}{L_{k}}\sum\limits_{v=1}^{L_{k}}s_v^{-2} \|{{\bm w'}}_v\|^2} \right)
    =P
\end{multline*}

\begin{equation}
    \lambda_i^{-1}= \frac{P}{L}+
    \frac{1}{L}\sum\limits_{l=1}^{L}\sigma^2s_l^{-2}\|{{\bm w'}}_{l}\|^2\left(\beta_k\ln\left(\frac{s_l^{-2} \|{{\bm w'}}_{l}\|^2}{\frac{1}{L_k}\sum\limits_{v \in \mathcal{L}_k}s_v^{-2} \|{{\bm w'}}_v\|^2} \right)+1\right)
    \label{lambd_calc}
\end{equation}

Substituting~\eqref{lambd_calc} into~\eqref{fk_calc} and~\eqref{fk_calc} into~\eqref{x_l} we get the required expressions for $x_l$ and then for $p_l$.


\end{document}